\documentclass[11pt,a4paper]{article}
\usepackage[T1]{fontenc}
\usepackage{lmodern}
\usepackage[a4paper,margin=25mm]{geometry}
\usepackage{amsmath,amssymb}
\usepackage{graphicx}
\usepackage{float}
\usepackage{booktabs,multirow,makecell}
\usepackage{microtype}
\usepackage{titlesec}
\usepackage{caption}
\usepackage[section]{placeins}
\usepackage{etoolbox}
\usepackage{mathtools} 
\usepackage[hidelinks]{hyperref}
\usepackage{bm}
\usepackage{color}
\newtheorem{thm}{Theorem}[section]

\newtheorem{lem}[thm]{Lemma}

\newtheorem{assum}[thm]{Assumption}
\newtheorem{prop}[thm]{Proposition}

\newtheorem{defn}[thm]{Definition}

\newtheorem{rem}[thm]{Remark}
\newtheorem{prob}[thm]{Problem}

\AtBeginEnvironment{table}{\small}
\titleformat{\section}{\fontsize{13}{16}\selectfont\bfseries}{\thesection}{0.65em}{}
\titleformat{\subsection}{\normalsize\bfseries}{\thesubsection}{0.65em}{}
\titlespacing*{\section}{0pt}{2.0ex plus .5ex minus .2ex}{0.9ex}
\titlespacing*{\subsection}{0pt}{1.5ex plus .3ex minus .2ex}{0.6ex}
\hypersetup{
  pdftitle={Control and Stability of a Multilevel Power System for a Future Distribution Network},
  pdfauthor={Xian Wu}
}
\begin{document}
\date{}
\title{\fontsize{18}{22}\selectfont
Control and Stability of a Multilevel Power System\\
for a Future Distribution Network}
\author{Xian Wu\thanks{
Xian Wu is with the Department of Mathematics, Zhejiang Sci-Tech University, Hangzhou, 310018, Zhejiang Province, China. She visited the Department of Applied Mathematics, Delft University of Technology, Delft, The Netherlands from November 2023 till October 2025. Xian Wu is grateful to the China Scholarship Council for financial support under Grant 202306220188 for her stay at Delft
University of Technology, The Netherlands(email: xianwu@zstu.edu.cn).}
\and
Jan H. van Schuppen\thanks{Jan H. van Schuppen is currently at Gouden Leeuw 143, 1103 KB Amsterdam, The Netherlands (email: vanschuppenjanh@freedom.nl).}
\and Hai Xiang Lin\thanks{Hai Xiang Lin is with the Department of Applied Mathematics, Delft University of Technology, Delft, 2628 CD, The Netherlands (email: H.X.Lin@tudelft.nl).}}
\maketitle
\begin{abstract}
  The growing integration of renewable energy sources into distribution networks poses significant challenges to frequency and voltage stability due to their intermittent nature and low-inertia dynamics. 
    This paper proposes a multilevel control framework for a future decarbonized power system, using energy storage systems as power buffers to mitigate frequency and voltage fluctuations. A nonlinear interconnected model is formulated to describe the multilevel aggregated dynamics of the distribution network.
    To reduce the operational complexity and communication overhead of the control system, a distributed linear quadratic regulator strategy is developed based on bottom-up information exchange, where each level implements local feedback control over a short time horizon. Local stability conditions are established for the open-loop nonlinear model, and the closed-loop stability is characterized by the Hurwitz property of the system matrix. Performance bounds are also derived to quantify the gap between the distributed and centralized control strategies.
\end{abstract}
\begin{quote}
\mbox{power distribution network}; \mbox{stability analysis}; \mbox{distributed control}.
\end{quote}
\section{Introduction}
\label{sec:introduction}
The increasing penetration of sustainable energy technologies is transforming conventional power systems into future decarbonized power grids. The distribution networks of these future power systems incorporate multiple renewable energy sources (RES), energy storage systems (ESS) and loads, all interacting within the same infrastructure \cite{Alobaidi,Panigrahi,dorfler2023}. 
While various modeling approaches have been developed to describe the physical interconnections and dynamic couplings of these components \cite{Tomsovic,Rick}, stability analysis of the resulting models is analytically difficult due to their nonlinear behavior. Additionally, the intermittent and low-inertia properties of RES introduce significant challenges to the frequency and voltage stability.
{To address these problems, this paper proposes a multilevel control framework that combines aggregated dynamic modeling, distributed LQR-based control design, and local stability analysis for a distribution network.}
\par
While ensuring asymptotic stability of the nonlinear system is essential, economic efficiency and technical feasibility of system operation are also important considerations. Consequently, the determination of the corresponding synchronous state is formulated as an optimal power flow (OPF) problem, which has been extensively studied for radial distribution networks under physical and operational constraints \cite{Low2013,Low2014,Molzahn,Lijun}.
A detailed analysis of various power flow models and their convex relaxations is presented in \cite{Low2013}, where structural properties and equivalence conditions are established. Two special power flow models and their relaxations are further discussed in \cite{Low2014}, emphasizing the importance of convex structures. 
{However, conventional OPF methods mainly focus on synchronous-state operation and economic dispatch, while transient deviations of frequencies and of voltages around the synchronous state are not explicitly modeled.}
\par
To reflect the state dynamics neglected in static frameworks, many studies focus on linearized models around a synchronous state. For example, a small-signal state-space model is developed for the current-source inverter of a photovoltaic generator to analyze its local dynamic behavior \cite{Rahman}. For low-inertia power systems, a full-order control problem is formulated for both synchronous and converter-based generators, incorporating device modeling, control schemes, as well as the dynamics of transmission lines and loads \cite{Markovic}. The interactions between local controllers and varying operating points critically affect the stability of the power system under various penetration levels of RES. These studies mainly analyze system dynamics under small-signal disturbances, but the transient dynamic response under rapid fluctuations of RES is not sufficiently addressed.
\par
These limitations motivate a detailed characterization of the transient dynamics of grid-forming converters connecting RES to the main grid \cite{Khan}.
A novel distributed voltage control method is designed to achieve reactive power sharing among converters \cite{Schiffer2014} under the necessary and sufficient conditions for stability. A Lyapunov characterization of almost global stability for dispatchable virtual oscillator control converters is obtained without line dynamics \cite{Colombino} and with line dynamics \cite{Dominic}, providing stability conditions on network parameters and control gains. Extensions to non-nominal synchronous states and to non-uniform networks have also been studied \cite{Xiuqiang}, yielding parametric conditions for the existence and the stability of synchronous states. 
\par
Although significant progress has been made in local converter control, recent studies focus on the coordination of grid-forming converters in distribution networks. 
A centralized control scheme has been proposed to coordinate converters and regulate voltage by collecting information from all photovoltaic converters \cite{Ciocia,Bidgoli}.
Because such an approach collects the states of all components, it suffers from severe limitations in scalability and robustness for large-scale distribution networks. These restrictions motivate the development of distributed control strategies, where each controller relies on local measurements and exchanges information with neighboring controllers. This structure reduces communication requirements and computational complexity, and enhances resilience to network failures.
\par
Among various distributed control approaches, multilevel control systems are designed based on the hierarchical structure of distribution networks. This control hierarchy originates from telephone networks \cite{Tanenbaum}, where the controller of each level operates independently and exchanges information with those at the lower level and at the next higher level. This architecture is widely used in small-scale distribution networks with two levels \cite{Plytaria,Kouveliotis,Sachs}, where controllers regulate local voltage and frequency and achieve coordination. 
\par
However, existing studies focus on the synchronous state, the small-signal stability of the linearized power system, or detailed modeling of converters for two-level distribution networks. {Few studies address multilevel architectures, local stability analysis and distributed control design for large-scale distribution networks.} Therefore, accurate mathematical modeling of the distribution network is still required to describe its interconnections and to formulate system dynamics for the stability analysis.
\par
This paper addresses these gaps by proposing a multilevel control framework for a future decarbonized power system. Specifically, we propose a six-level control system for a large-scale distribution network, integrating RES at each level and ESSs at several levels. 
These ESSs provide control inputs to regulate frequency and voltage at higher levels. At discrete time intervals ({for example}, every five minutes), aggregated power demand information is communicated from a lower level to its next higher level. Each level uses feedback control based only on its own state to regulate frequency and voltage within a short time horizon. The main contributions of this paper are summarized as follows. 
\begin{enumerate}
\item {A multilevel control system of a large-scale distribution network with aggregated RESs, ESSs and loads, where ESSs provide inputs for frequency and voltage regulation.}
\item {A Lyapunov-based local stability analysis of the open-loop combined power system.}
\item {A distributed LQR-based control design for an isolated power system, with closed-loop stability analysis and performance comparison based on controllability and observability.}
\end{enumerate}
\par
Further details of the numerical examples related to this paper are provided in the report of the examples and computations.
\par
The remainder of the paper is organized as follows. 
In Section \ref{sec:problem}, we formulate the problem of a distribution network for a future power system and introduce the control objectives of a multilevel power system. Section \ref{sec:multictrl} introduces the multilevel system structure, followed by Section \ref{sec:high-level power system}, which provides the detailed dynamic modeling for Levels 3, 4, and 5. {Section \ref{sec:open_loop} characterizes the local stability of the synchronous-state trajectory for the open-loop system through Lyapunov-based analysis. Section \ref{sec:distributed_control} designs a distributed linear quadratic regulator (LQR) controller using only local state information at each level of the isolated model. 
Since these isolated models neglect the dynamic couplings between adjacent levels, Section \ref{sec:closed-loop} analyzes the combined power system and provides conditions for closed-loop stability. Section \ref{sec:performance} compares the performance of the distributed and centralized control methods. 
}Finally, conclusions are provided in Section \ref{sec:conclusion}.


\section{Problem formulation}\label{sec:problem}
A power system consists of a transmission network and multiple distribution networks. Traditionally, power flows from the transmission network to the distribution networks through their points of common coupling (PCC). In future power systems, advanced control strategies will be required for distribution networks that supply power directly to end users, due to the increasing integration of distributed energy resources in distribution networks.
\par
This paper is motivated by the transition of distribution networks from conventional generation to RES, such as solar, wind, and biomass generators. These sources raise major concerns about the stability of future power systems due to their intermittent nature and low-inertia properties. 
Therefore, it is necessary to design a control framework that uses large-capacity ESSs as buffers to mitigate frequency and voltage fluctuations in large-scale distribution networks.
\par
\begin{prob}{\em Problem formulation for the distribution network of a future power system.}
 \begin{enumerate}
    \item How to formulate a multilevel control system for a distribution network of a future power system?
    \item How to synthesize controllers of the combined power system at higher levels of the multilevel control system?
  \end{enumerate}
\end{prob}
\par
{To address these challenges, we consider a multilevel control framework in which the overall time horizon is partitioned into a sequence of five-minute periods. Within each period, the power flows between adjacent levels are approximated by quasi-static values obtained from OPF-based scheduling methods, while the instantaneous disturbances caused by the dynamics of frequencies and voltages are not explicitly represented. Under this assumption, the control objectives of each level are specified systematically.}
\begin{defn}{\em Control objectives of a multilevel power system.} 
    \begin{enumerate}
        \item Ensure that the power supply equals the predicted power demand for every five minutes.
        \item Reduce the variance of voltages and of frequencies at each level within every five minutes.
    \end{enumerate}
\end{defn}
The first objective ensures power balance by scheduling the power supply to meet the predicted demand at each level. 
The second objective requires coordinated control methods across multiple levels, where each controller relies on local measurements and on the aggregated power demand of other controllers.
{Consequently, the fluctuations of the power flows are reduced in size of the power flows from 
the distribution network to the low-inertia transmission network.}
\section{Multilevel power system of a distribution network}\label{sec:multictrl}
{This section presents a six-level hierarchical framework of a radial distribution network incorporating RESs, ESSs, and loads. This framework combines the aggregated demand of Levels 0--2 with the frequency and voltage dynamics of Levels 3--5, and integrates this multilevel structure with Lyapunov-based stability analysis and distributed LQR-based control design.}
\par
Multilevel systems have a long history in control theory and networked systems \cite{mesarovic:macko:takahara:1970,komenda:masopust:schuppen:2016:wodes,findeisen:bailey:brdys:malinowski:tatjewski:wozniak:1980,kempker:ran:schuppen:2014:ieeetac,komenda:masopust:schuppen:2015:deds}. Similar hierarchical structures also appear in telephone and communication networks \cite{Tanenbaum}, where levels are often organized according to geographical extent.
\par
Distributed controllers improve scalability and reduce the need for centralized information processing. However, the stability and performance of the overall closed-loop system depend on the interactions among controllers at different levels.
Therefore, we propose a six-level architecture for a distribution network based
on geographical extent and the estimated order of magnitude of power demand.
This architecture is intended to provide a hierarchical framework for the distributed control strategy, rather than to replace the conventional subdivisions of the power system. At each level, local information is used to regulate frequency and voltage, with closed-loop stability and performance as the main control objectives. 
\par
The six levels are defined as follows, and the corresponding number of inhabitants is listed for illustration purposes.
\begin{defn}{\em Levels of the multilevel control system for a distribution network.}
\begin{itemize}
  \item {\em Level 0.} An electric building. There are several types of electric buildings, including urban houses, farms, shops, factories, office buildings, and small, medium, and large industrial companies. 
  \item {\em Level 1.} A street with 100 electric buildings. This street is assumed to accommodate 300 inhabitants if three inhabitants live in each building, although this number can vary between 200 and 500.
  \item {\em Level 2.} A neighborhood with 10 streets. The neighborhood is assumed to have 3,000 inhabitants.
  \item {\em Level 3.} A town with 10 neighborhoods. The town is assumed to have 30,000 inhabitants.
  \item {\em Level 4.} A city with 10 towns. The city is assumed to have 300,000 inhabitants. 
  \item {\em Level 5.} A region with 10 cities. The region is assumed to have 3,000,000 inhabitants.
\end{itemize}
\end{defn}
To characterize the physical interconnections of this multilevel structure, the radial distribution network is modeled as a tree graph $\mathcal{G}=(\mathcal{V},\mathcal{E})$, where $\mathcal{V}$ is the set of $n$ nodes and $\mathcal{E}\subset\mathcal{V}\times\mathcal{V}$ denotes the set of lines. In this graph, buses correspond to nodes and power lines correspond to branches. Each node is connected to its unique parent and a set of children, except for the PCC and the electric buildings. The PCC serves as the root of the radial distribution network, while buildings are modeled as leaf nodes.
\par
We denote by $\mathcal{L} = \{0,1,\dots,5\}$ the set of level indices, and by $\mathcal{L}_a$ the set of nodes at Level $a\in\mathcal{L}$.
Therefore, the set of all nodes is
\begingroup
\setlength{\abovedisplayskip}{3pt}
\setlength{\belowdisplayskip}{3pt}
    \begin{align*}
        \mathcal{V} = \bigcup_{a\in\mathcal{L}}  \mathcal{L}_a,\quad \mathcal{L}_a\bigcap\mathcal{L}_b=\emptyset, \quad \forall a,b\in\mathcal{L},\ a\neq b.
    \end{align*}
\endgroup
{The network is assumed to contain only inter-level lines connecting nodes in adjacent levels. Direct electrical connections between nodes at the same level are not included in our model. For each node $i\in\mathcal{L}_a$, $a\in\mathcal{L}$, we define
\begingroup
\setlength{\abovedisplayskip}{3pt}
\setlength{\belowdisplayskip}{3pt}
\begin{align*}
\mathcal{N}_i^{\text{high}}&=\{j\in\mathcal{L}_{a+1}\mid (i,j)\in\mathcal{E}\},\\
\mathcal{N}_i^{\text{low}}&=\{j\in\mathcal{L}_{a-1}\mid (i,j)\in\mathcal{E}\},
\end{align*}
\endgroup
where $\mathcal{N}_i^{\mathrm{high}}$ and $\mathcal{N}_i^{\mathrm{low}}$ denote the neighboring nodes at the higher levels and lower levels of node $i$ respectively. In the radial distribution network, $\mathcal{N}_i^{\mathrm{high}}$ contains the unique parent node of $i$, except for the PCC, while $\mathcal{N}_i^{\mathrm{low}}$ contains the child nodes of $i$, except for leaf nodes.
\par
Although practical distribution networks may have bidirectional power flows due to RES injection, we consider a scenario where power from RES is supplied to loads or stored in ESS, with no power injection from buildings to the grid. Therefore, power is transferred from higher levels to lower levels, while predicted power demands are communicated from lower levels to higher levels. Then the controllers of higher levels schedule the power supply to their lower levels.
\par
Within this framework, Levels 0, 1, and 2 are represented by their aggregated predicted power demand and quasi-static power flows with higher levels, while their detailed load dynamics are neglected. Therefore, this paper focuses on the modeling and control synthesis of Levels $\mathcal{L}_H = \{3,4,5\}$, Level 5, Level 4, and Level 3 contain 1, 10, and 100 nodes respectively. Thus, $n_H = |\mathcal{V}_H|=111$, where $\mathcal{V}_H=\bigcup_{a \in \mathcal{L}_H} \mathcal{L}_a$.}
The multilevel structure and notation defined above provide the analytical foundation for the dynamic modeling presented in the following section.
\section{Power system of Levels 3, 4, and 5}\label{sec:high-level power system}
{
At each level, local sources, loads and ESSs are connected to a Level bus. The voltage magnitudes at different levels may differ, as is typical in practical distribution networks. The transformers connecting these voltage levels are not explicitly included in the model, as they are assumed to be lossless. Level 5 includes a nuclear power plant, a solar park, a wind park, and a biomass generator. Its total power demand consists of the local loads and the aggregated demand of its child nodes at Level 4. 
\par
Conventional generators, including the nuclear power plant and the biomass generator, are modeled by classical second-order swing equations with physical inertia. ESS-supported units for frequency and voltage regulation are represented as grid-forming converters providing virtual inertia and damping. Other inverter-based renewable energy sources are connected to the grid through grid-following converters. Their current-controller and phase-locked-loop dynamics are neglected because they evolve on a much faster time scale than the dynamics of phase angles, frequencies and voltages. The models of Levels 3 and 4 are analogous to that of Level 5, with differences in inertia and power sources.
\par
The system is considered on three time scales. The prediction and scheduling horizon is divided into five-minute intervals, within which the {\em scheduled values} of the input power flows are constant according to predictions for the next five minutes. On the 5--300 s time scale, small fluctuations from the scheduled values arise from load switching, ESS charging, and variations of RES generation, and the stability of the open-loop system will be investigated. The faster dynamics of grid-forming converters, on the 20--100 ms time scale, are not explicitly modeled and are represented by {\em quasi-static values} during the five-minute interval.}
\par
We assume that the phase angle, the frequency, and the voltage of all power sources and loads at each level are synchronized at the corresponding Level bus within a five-minute period.
{
Each Level bus is represented by an aggregated model incorporating the effects of synchronous generators and their local control loops through aggregated inertia, damping coefficients, and scheduled values of the input power flows. The buses of Levels 3, 4, and 5 are connected, so power fluctuations at one level affect the others. The aggregated dynamics are analogous to those of nodes of transmission networks\cite{dorfler2023,Venkatasubramanian}. However, the aggregated model is developed for a multilevel distribution network, where ESSs provide inputs for frequency and voltage regulation.}
\begin{defn}{\em Combined power system of the Level buses for Levels 3, 4, and 5.} 
{For each level bus $i\in\mathcal{L}_a,~a\in\mathcal{L}_H$, let $(\theta_i,\omega_i,v_i)\in\mathbb{T}\times\mathbb{R}\times\mathbb{R}_{>0}$ denote its phase angle, 
the frequency deviation from the nominal frequency, and the voltage magnitude, where $\mathbb{T}:=\mathbb{R}/2\pi\mathbb{Z}$. The aggregated dynamics of Levels 3, 4, and 5 are defined with respect to a reference rotating frame at the nominal system frequency as follows.}
\begingroup 
\setlength{\parskip}{3pt}
\setlength{\abovedisplayskip}{3pt}
\setlength{\belowdisplayskip}{3pt}
     \begin{align} \label{combined_eq} \frac{\text{d}\theta_i(t)}{\text{d}t}&=\omega_i(t)\nonumber,\\ 
        m_i\frac{\text{d}\omega_i(t)}{\text{d}t}&=-d_i\omega_i(t)+P_i-\sum_{j\in\mathcal{V}_H} P_{i,j}(\theta,v)+u_{i,\omega}(t),\nonumber\\ 
        \tau_i\frac{\text{d}v_i(t)}{\text{d}t}&=-k_iv_i(t)+Q_i-\sum_{j\in\mathcal{V}_H} Q_{i,j}(\theta,v)+u_{i,v}(t),\nonumber\\ 
    P_i&=P_{sp,i}-P_{ld,i},\\ Q_i&=Q_{sp,i}-Q_{ld,i},\nonumber\\ 
    P_{sp,i}&=P_{sol,i}+P_{wind,i}+P_{bm,i},\nonumber\\ 
    Q_{sp,i}&=Q_{sol,i}+Q_{wind,i}+Q_{bm,i},\nonumber\\ 
    P_{i,j}(\theta,v)&=v_iv_j\left[g_{ij}\cos(\theta_i-\theta_j)+b_{ij}\sin(\theta_i-\theta_j)\right],\nonumber\\ 
    Q_{i,j}(\theta,v)&=v_iv_j\left[g_{ij}\sin(\theta_i-\theta_j)-b_{ij}\cos(\theta_i-\theta_j)\right].\nonumber \end{align}
\endgroup
    Here $m_i$ and $d_i$ denote the aggregated inertia and damping coefficients associated with the frequency dynamics, while $\tau_i$ and $k_i$ are the time constant and damping coefficient associated with the voltage dynamics. 
    {Note that, due to the lower penetration of synchronous machines at lower levels \cite{Pagnier}, the aggregated inertia decreases from higher levels to lower levels. Therefore, these coefficients satisfy
$
    m_{i} > m_{j} > m_{k},~\forall i\in\mathcal{L}_5,j\in\mathcal{L}_4,k\in\mathcal{L}_3.
$}
\par
    Denote by $P_{sol,i}$, $P_{wind,i}$, $P_{bm,i}$ the predicted active power supply from the solar park, the wind park and the biomass generator respectively, and by $P_{ld,i}$ the predicted active power load of Level $i$ over a five-minute period. Similarly, denote by $Q_{sol,i}$, $Q_{wind,i}$, $Q_{bm,i}$, and $Q_{ld,i}$ the corresponding reactive power of these sources and loads. For $i\in \mathcal{L}_5$, $P_{sp,i}$ and $Q_{sp,i}$ also include the power supply from the nuclear power plant. \par 
    Denote by $P_{i,j}$ and $Q_{i,j}$ the dynamic active and reactive power flows between adjacent levels, which are described by phase angle differences and voltage magnitudes. $g_{ij}$ and $b_{ij}$ are the conductance and susceptance between the Level bus $i$ and $j$, respectively. $u_{i,\omega}(t)$ and $u_{i,v}(t)$ denote the active power and reactive power inputs of the ESS.
\end{defn}
\begin{rem}
{In the combined power system \eqref{combined_eq}, all variables are expressed in per unit, and the system is analyzed around the nominal voltage $v^n=1$ p.u.
Let $v_i = v^n + \Delta v_i$. Then, approximating $v_i$ by  $v^n$, we obtain
$1/v_i\approx 1/v^n$, where terms involving $\Delta v_i$ are neglected since such deviations are assumed to be small around the synchronous state. Therefore, the term $(Q_i-\sum_{j\in\mathcal{V}_H} Q_{i,j}(\theta,v))/{v_i}$ is approximated by $Q_i-\sum_{j\in\mathcal{V}_H} Q_{i,j}(\theta,v)$ in the neighborhood of the nominal voltage.}
\end{rem}
\section{Stability analysis of the open-loop power system}\label{sec:open_loop}
A power system can be modeled as a nonlinear dynamic system, whose behavior is characterized by the properties of its synchronous state. These properties include:
\begin{enumerate}
    \item Existence of a synchronous state \cite{Kumagai},
    \item The local asymptotic and transient stability of the corresponding synchronous-state trajectory, 
    \item The domain of attraction associated with the synchronous-state trajectory.
\end{enumerate}
\par
Although the existence of a synchronous state and its domain of attraction are important for understanding the overall system behavior, this section focuses on the local asymptotic stability of the open-loop system without ESS inputs. {We also characterize the local behavior of the aggregated nonlinear dynamics and establish the deviation coordinates used for the closed-loop stability analysis in Section \ref{sec:closed-loop}.}
For an inner approximation of the domain of attraction, see \cite{Vannelli}.
\par
Here, we consider the {open-loop combined power system} by setting the control inputs to zero, i.e., $u_{i,\omega}=0$ and $u_{i,v}=0$ in \eqref{combined_eq}.
Analyzing the stability of this system reveals its inherent dynamic properties and provides guidelines for designing distributed controllers in the subsequent section. 
\par 
Let ${\bm x}=\text{col}({\bm x_i})\in(\mathbb{T}\times\mathbb{R}\times\mathbb{R})^{n_H}$ denote the state vector, where ${\bm x_i}=({\theta_i},\omega_i,{v_i})^\top$. 
In the following, we define the synchronous state of the power system.
{\begin{defn}{\em Synchronous state of the power system.} 
\par 
A vector $\bm x^*=\big((\theta_i^*,\omega_i^*,v_i^*))$ 
is called a synchronous state if all frequencies are identical, i.e., $\omega_i^*=\omega^*$ for all $i\in\mathcal{V}_H$, and the following equations hold.
\begingroup
\setlength{\abovedisplayskip}{3pt}
\setlength{\belowdisplayskip}{3pt}
\begin{equation}\label{synstate} 
    \begin{aligned} 
        0&=-d_i\omega_i^*+P_i-\sum_{j\in\mathcal{V}_H} P_{i,j}(\theta^*,v^*),\\
        0&=-k_iv_i^*+Q_i-\sum_{j\in\mathcal{V}_H} Q_{i,j}(\theta^*,v^*).
    \end{aligned} 
\end{equation}
\endgroup
\end{defn}
We assume there exists a synchronous state $\bm x^*$ satisfying the above conditions. }
A synchronous-state trajectory is defined as 
\begingroup
\setlength{\abovedisplayskip}{3pt}
\setlength{\belowdisplayskip}{3pt}
\begin{align*}\bm x^s(t)
=
\operatorname{col}\bigl(
(\theta_i^s(t),\omega_i^s(t),v_i^s(t))^\top
\bigr),
\end{align*}
\endgroup 
where $\theta_{i}^s(t)=\text{mod}_{2\pi}(\theta_i^*+t\omega^*),~\omega_{i}^s(t)=\omega^*,~v_{i}^s(t)=v_i^*$.
\begin{defn}{\em Deviations of the state variables.} 
    \par
    Select a reference node $g\in \mathcal{V}_H$ to remove the rotational degree of freedom. For each $i\in\mathcal{V}_H\setminus\{g\}$,
define
\begingroup
\setlength{\abovedisplayskip}{3pt}
\setlength{\belowdisplayskip}{3pt}
\begin{align*}
  \widetilde{\theta}_i(t)=[\theta_i(t)-\theta_g(t)]-[\theta_{i}^*-\theta_g^*],  
\end{align*}
\endgroup
and set $\widetilde{\theta}_g(t)=0$.
\\
For all $i\in\mathcal{V}_H$, define
\begingroup
\setlength{\abovedisplayskip}{3pt}
\setlength{\belowdisplayskip}{3pt}
\begin{align*}
\widetilde{\omega}_i(t)=\omega_i(t)-\omega_{i}^*,~~
\widetilde v_i(t)=v_i(t)-v_{i}^*.
\end{align*}
\endgroup
The reduced deviation vector relative to the synchronous-state trajectory is defined as
\begingroup
\setlength{\abovedisplayskip}{3pt}
\setlength{\belowdisplayskip}{3pt}
\begin{equation*}
  \widetilde{\bm x} = \text{col}(\widetilde{\bm x}_g,(\widetilde{\bm x}_i)_{i\in\mathcal{I}})\in\mathbb T^{n_H-1}\times\mathbb R^{n_H}\times\mathbb R^{n_H}, 
\end{equation*} 
\endgroup
where
$
  \widetilde{\bm x}_g=(\widetilde{\omega}_g,\widetilde{v}_g)^\top,~~\widetilde{\bm x}_i=(\widetilde{\theta}_i,\widetilde{\omega}_i,\widetilde{v}_i)^\top,~ i\neq g.  $
  Here, the phase angle deviations are defined based on the reference node, while the frequency and voltage deviations are defined based on the synchronous-state trajectory. {Throughout the remainder of this paper, we use $\mathcal{I}$ to denote $\mathcal{V}_H\backslash\{g\}$.}
\end{defn}
Assume that, after fixing the reference angle, \eqref{synstate} has a locally unique solution $\bm x^*$.
Based on the above definitions, 
the resulting autonomous system \eqref{combined_eq} can be expressed in the following equations,
\begingroup
\setlength{\abovedisplayskip}{3pt}
\setlength{\belowdisplayskip}{3pt}
\begin{equation}\label{eq_nonlinear_com}
    \begin{aligned}
        \frac{\mathrm{d}\widetilde{\theta}_i(t)}{\mathrm{d}t}&=\widetilde{\omega}_i(t)-\widetilde{\omega}_g(t),\notag\\
        m_i\frac{\mathrm{d}\widetilde{\omega}_i(t)}{\mathrm{d}t}&=-d_i\widetilde{\omega}_i(t)+\sum_{j\in\mathcal{V}_H}P_{i,j}(\theta^*,v^*)-\sum_{j\in\mathcal{V}_H}P_{i,j}(\theta,v),\\
        \tau_i\frac{\mathrm{d}\widetilde{v}_i(t)}{\mathrm{d}t}&=-k_i\widetilde{v}_i(t)+\sum_{j\in\mathcal{V}_H}Q_{i,j}(\theta^*,v^*)-\sum_{j\in\mathcal{V}_H}Q_{i,j}(\theta,v).\notag
    \end{aligned}
\end{equation}
\endgroup
The compact vector form of the above equation is
\begingroup
\setlength{\abovedisplayskip}{3pt}
\setlength{\belowdisplayskip}{3pt}
\begin{equation}\label{eq_nonlinear}
    \begin{aligned}
         \frac{\mathrm{d}\widetilde{\bm x}(t)}{\mathrm{d}t} &= \bm f(\widetilde{\bm x}(t)).
    \end{aligned}
\end{equation}
\endgroup
Here, $\bm f:\mathbb{T}^{n_H-1}\times \mathbb{R}^{n_H}\times\mathbb R^{n_H}
\to\mathbb{R}^{3n_H-1}$ is the nonlinear vector-valued function induced by \eqref{combined_eq}. 
We denote by $\widetilde{\bm x}(t;0,\widetilde{\bm x}(0))$ the solution
of \eqref{eq_nonlinear} at time $t$ with initial condition
$\widetilde{\bm x}(0)$ at time $0$.
The equality $\bm f(\bm 0)=\bm 0$ follows from \eqref{synstate}. Define a new vector $\bar{\theta}_i=\theta_i^*+\widetilde{\theta}_i$, while $\bar{\theta}_g=\theta_g^*$. The phase angle differences satisfy
$\bar{\theta}_i-\bar{\theta}_j=\theta_i-\theta_j$.
\par
Before analyzing the stability of the synchronous-state trajectory, the following assumptions are listed for our power system.
\begin{assum}{\em {Assumptions for local stability analysis of the open-loop system at Levels 3, 4, and 5.}}\label{assump:system}
    Consider the combined power system of higher levels and its synchronous state $\bm x^*$. 
\begin{enumerate}
    \item \textbf{Network properties.}
    The power network $\mathcal{G}=(\mathcal{V},\mathcal{E})$ is connected and lossless, i.e.,
\begingroup
\setlength{\abovedisplayskip}{3pt}
\setlength{\belowdisplayskip}{3pt}
\begin{align*}
   b_{ij} = b_{ji} > 0, \quad \forall (i,j) \in \mathcal E, \ i \neq j,  
\end{align*}
\endgroup
    and 
  \begingroup
\setlength{\abovedisplayskip}{3pt}
\setlength{\belowdisplayskip}{3pt}
\begin{align*}
     g_{ij} = 0, \quad \forall i,j \in \mathcal V.
\end{align*}
\endgroup
    \item \textbf{Synchronous angle constraints.}
    The synchronous-state trajectory of phase angle satisfies 
   \begingroup
\setlength{\abovedisplayskip}{3pt}
\setlength{\belowdisplayskip}{3pt}
\begin{align*}
        |\theta_{ij}^*|=|\theta_{i}^*-\theta_{j}^*|< \pi/2, \quad \forall (i,j)\in\mathcal E.
   \end{align*}
\endgroup
    \item \textbf{System parameters.}
      \begingroup
\setlength{\abovedisplayskip}{3pt}
\setlength{\belowdisplayskip}{3pt}
\begin{align*}
    m_i,d_i,\tau_i,k_i,v_i^*>0,~~ \forall i\in \mathcal{V}_H.
     \end{align*}
\endgroup
    \item \textbf{Matrix properties.} Consider the synchronous-state trajectory $\bm x^*$, define 
    \begin{equation*}
        \begin{aligned}
            \bm V_{\omega\omega} &= \bm M = \mathrm{diag}(m_1, \dots, m_n),\\
             \bm V_{\theta\theta}(i,j)&=\begin{cases}
                \sum_{(j,i)\in\mathcal{E}} b_{ij}v_{i}^*v_j^*\cos\theta_{ij}^*, & i=j\\
                -b_{ij}v_i^*v_j^*\cos\theta_{ij}^*,& i\neq j
             \end{cases},\\
                \bm V_{vv}(i,j)&=\begin{cases}
                    \frac{Q_i}{(v_i^*)^2}-b_{ii},&i=j\\
                -b_{ij}\cos\theta_{ij}^*,& i\neq j
                \end{cases},\\
            \bm V_{\theta v}(i,j)&=\begin{cases}
                \sum_{(j,i)\in\mathcal{E}} b_{ij}v_j^*\sin\theta_{ij}^*,&i=j\\
                b_{ij}v_i^*\sin\theta_{ij}^*,& i\neq j
            \end{cases},
        \end{aligned}
    \end{equation*}
    where
          \begingroup
\setlength{\abovedisplayskip}{3pt}
\setlength{\belowdisplayskip}{3pt}
\begin{align*}\bm V_{\omega\omega},\bm V_{vv}\in\mathbb{R}^{{n_H}\times{n_H}}, \bm V_{\theta\theta}\in\mathbb{R}^{(n_H-1)\times(n_H-1)},
\end{align*}
\endgroup
    and       \begingroup
\setlength{\abovedisplayskip}{3pt}
\setlength{\belowdisplayskip}{3pt}
\begin{align*}\bm V_{\theta v}\in\mathbb{R}^{(n_H-1)\times n_H}.\end{align*}
\endgroup
    Under assumption (3), \(\bm V_{\omega\omega} \succ 0\). In addition, by Lemma \ref{lem:Vtheta} and Lemma \ref{lem:Vv}, $\bm V_{\theta\theta}$ and $\bm V_{vv}$ are positive definite, hence the smallest eigenvalue satisfies \(\lambda_{\min}(\bm V_{\theta\theta}) > 0\).
    \item \textbf{Candidate Lyapunov function.} Define the candidate Lyapunov function $V: \mathbb{R}^{3n_H-1}\to\mathbb{R}$ as
    \begin{equation*}
  \scalebox{1}{$
        \begin{aligned}
            V(\widetilde{\bm x})&=\sum_{i\in\mathcal{V}_H} \frac{1}{2}m_i\widetilde{\omega}_i^2+U(\bm\bar{\theta},\bm v^*\!+\! \widetilde{\bm v}) - U(\bm \theta^*,\bm v^*),\\
            U(\bm\theta,\bm v)&=\sum_{i\in\mathcal{V}_H}\bigg[-(P_i-d_i\omega_i^*)\theta_i+k_iv_i-Q_i\ln v_i-\sum_{j\in\mathcal{N}_i^{\text{high}}}\! b_{ij}v_iv_j\cos\theta_{ij}-\frac{1}{2}b_{ii}v_i^2\bigg].
                  \end{aligned}$}
    \end{equation*}
This candidate Lyapunov function extends the energy function in \cite{Varaiya1985} with additional voltage dynamics.
It follows that \(V\in \mathcal{C}^2\). Assume that there exists a constant \(L > 0\), such that for all $\bm x,\bm y$ in the bounded domain defined in Theorem \ref{thm:stability},
   \begin{equation*}
    \begin{aligned}
        &\|\nabla^2 V(\bm x)-\nabla^2 V(\bm y)\|\leq L\|\bm x-\bm y\|.
    \end{aligned}
\end{equation*}
    \item \textbf{Stability conditions.} Define
\begingroup
\setlength{\abovedisplayskip}{3pt}
\setlength{\belowdisplayskip}{3pt}
\[
\scalebox{1}{$\displaystyle
  C_1=\sum_{i\in\mathcal{V}_H}\Big[\big(\sum_{(j,i)\in\mathcal{E}} b_{ij}v_j^*\sin\theta_{ij}^*\big)^2
      +\mkern -10mu \sum_{(j,i)\in\mathcal{E}}(b_{ij}v_i^*\sin\theta_{ij}^*)^2\Big],
$}
\]
\endgroup
    and
          \begingroup
\setlength{\abovedisplayskip}{3pt}
\setlength{\belowdisplayskip}{3pt}
\begin{align*}
        C_2= \min_{i\in\mathcal{V}_H}\bigg(\frac{Q_i}{(v_i^*)^2}-\mkern -10mu\sum_{j\in\mathcal{V}_H} b_{ij}\cos\theta_{ij}^*\bigg)\lambda_{\min}\big(\bm V_{\theta\theta}\big).
\end{align*}
\endgroup
    Assume that $C_2>C_1\ge 0$.
    \item \textbf{Local stability region.} 
    \par
Choose $r>0$ sufficiently small such that the synchronous state is locally unique in the neighborhood
\begin{equation*}
    \Omega_r=\{\widetilde{\bm x} \in \mathbb{T}^{n_H-1}\times \mathbb{R}^{n_H}\times\mathbb R^{n_H}\mid \|\widetilde{\bm x}\| \le r \},
\end{equation*}
To ensure that the Lyapunov function has a quadratic bound in $\Omega_r$, we further let $r < 3 \lambda_{\min}(\nabla^2 V(\bm 0))/L$. Define

\begin{equation*}
\begin{aligned}
    c_1 &= \frac12 \lambda_{\min}(\nabla^2 V(\bm 0)) - \frac{L r}{6},\\
    c_2 &= \frac12 \lambda_{\max}(\nabla^2 V(\bm 0)) + \frac{L r}{6},\\
    \delta(\epsilon) &= {\sqrt{\frac{c_1}{2c_2}} \min\{\epsilon, r\}.}
\end{aligned}
\end{equation*}
where 
 \begin{equation*}
    \begin{aligned}
       \lambda_{\min}(\nabla^2 V(\bm 0))&=\min\{\lambda\in\mathbb{R}|\lambda\in\text{spec}(\nabla^2 V(\bm 0))\},\\
     \lambda_{\max}(\nabla^2 V(\bm 0))&=\max\{\lambda\in\mathbb{R}|\lambda\in\text{spec}(\nabla^2 V(\bm 0))\}.
    \end{aligned}
 \end{equation*}
 The choice of $r$ implies that $c_1>0$, see \eqref{c10}. By Lemma \ref{lem:Vtheta}, Lemma \ref{lem:G} and parameter assumptions, we will prove that $\lambda_{\min}(\nabla^2 V(\bm 0))>0$, see Theorem \ref{thm:stability}.
\end{enumerate}
\end{assum}
\begin{rem}
    {
Assumption \ref{assump:system} introduces sufficient conditions for the Lyapunov-based analysis of Theorem \ref{thm:stability}.
The lossless assumption is used for the aggregated model of Levels 3, 4, and 5.
The other assumptions ensure the block structure and the positive definiteness of the Hessian matrix of the Lyapunov function. 
Because these assumptions are introduced to facilitate the stability analysis, they should be regarded as sufficient conditions rather than physical constraints. As a result, they are not required for the distributed controller design. Their conservativeness mainly results from the analytical techniques used in the proof.}
\end{rem}
\par
{The following three lemmas verify the positive definiteness of the Hessian matrix, which is used for the Lyapunov-based stability analysis.}
\begin{lem}\label{lem:Vtheta}
Under Assumption~\ref{assump:system}, the matrix $\bm V_{\theta\theta}$ is symmetric and positive definite. 
Hence, all its eigenvalues are real and positive.
\end{lem}
\textit{{Proof: }}
     Under Assumption \ref{assump:system} (1)-(3), the matrix $\bm V_{\theta\theta}$ is symmetric, with the following elements,
     \begingroup
     \setlength{\abovedisplayskip}{3pt}
\setlength{\belowdisplayskip}{4pt}
        \begin{align}\label{hessiantheta}
            \bm V_{\theta\theta}(i,j)&=-b_{ij}v_i^*v_j^*\cos\theta_{ij}^*<0,~ (i,j)\in\mathcal{E},\notag\\
          \bm V_{\theta\theta}(i,i)&=\sum_{(j,i)\in\mathcal{E}} b_{ij}v_i^*v_j^*\cos\theta_{ij}^*=-\sum_{\substack{j\in\mathcal{I}\\j\ne i}}\bm V_{\theta\theta}(i,j)+b_{ig}v_i^*v_g^*\cos\theta_{ig}^*>0.
        \end{align}    
    \endgroup
    For any $\bm y\in \mathbb{R}^{n_H-1}$, the quadratic form of $\bm V_{\theta\theta}$ is obtained,
    \begingroup
    \setlength{\abovedisplayskip}{3pt}
\setlength{\belowdisplayskip}{3pt}
        \begin{align*}
           \bm y^\top \bm V_{\theta\theta}\bm y&=
             \sum_{i\in \mathcal{I}}\biggl[\sum_{\substack{j\in\mathcal{I}\\j\ne i}} \Bigl(\bm V_{\theta\theta}(i,j)y_iy_j\!-\!\frac{1}{2}\bm V_{\theta\theta}(i,j)y_i^2-\frac{1}{2}\bm V_{\theta\theta}(i,j)y_j^2\Bigr)+b_{ig}v_i^*v_g^*\cos\theta_{ig}^*y_i^2\biggr].
        \end{align*}
    \endgroup
    We also have
    \begingroup
        \setlength{\abovedisplayskip}{3pt}
\setlength{\belowdisplayskip}{3pt}
    \begin{equation*}
        \begin{aligned}
          \sum_{i\in \mathcal{I}}\sum_{\substack{j\in\mathcal{I}\\j\ne i}}\bm V_{\theta\theta}(i,j)y_i^2
          =\sum_{j\in \mathcal{I}}\sum_{\substack{i\in\mathcal{I}\\i\ne j}}\bm V_{\theta\theta}(j,i)y_j^2= \sum_{i\in \mathcal{I}}\sum_{\substack{j\in\mathcal{I}\\j\ne i}}\bm V_{\theta\theta}(j,i)y_j^2= \sum_{i\in \mathcal{I}}\sum_{\substack{j\in\mathcal{I}\\j\ne i}}\bm V_{\theta\theta}(i,j)y_j^2. 
        \end{aligned}
    \end{equation*}
    \endgroup
    The first equality holds by exchanging the indices $i$ and $j$, and the second equality follows from the exchange of the summation indices. Finally, the symmetry property of the matrix $\bm V_{\theta\theta}$ yields the last equality. Then we obtain
              \begingroup
\setlength{\abovedisplayskip}{3pt}
\setlength{\belowdisplayskip}{3pt}
\begin{align*}
             \bm y^\top \bm V_{\theta\theta}\bm y&=
             \sum_{i\in \mathcal{I}}\biggl[\sum_{\substack{j\in\mathcal{I}\\j\ne i}} \Bigl( \bm V_{\theta\theta}(i,j)y_iy_j\!-\!\frac{1}{2}\bm V_{\theta\theta}(i,j)y_i^2-\frac{1}{2}\bm V_{\theta\theta}(i,j)y_j^2\Bigr)+b_{ig}v_i^*v_g^*\cos\theta_{ig}^*y_i^2\biggr]\\
             &\!=\!\sum_{i\in \mathcal{I}}\Bigl[
                -\frac{1}{2}\sum_{\substack{j\in\mathcal{I}\\j\ne i}} \bm V_{\theta\theta}(i,j)(y_i\!-\!y_j)^2\!+\!b_{ig}v_i^*v_g^*\cos\theta_{ig}^*y_i^2
             \Bigr]\\
             &\ge 0. 
        \end{align*}
    \endgroup
    If the quadratic form is zero, then \(y_i=y_j\) for every pair of adjacent non-reference nodes. Since $\mathcal{G}$ is connected, every $i\in\mathcal{V}_H$ is connected to $g$ by a path, then \(y_i=0\) for
    the node connected to \(g\). These conditions imply
\(\bm y=\bm0\).
  Thus $\bm V_{\theta\theta}$ is a positive definite matrix, whose real eigenvalues are positive.
    \hfill $\square$
\par
\begin{lem}\label{lem:Vv}
Under Assumption~\ref{assump:system}, the matrix $\bm V_{vv}$ is positive definite.
\end{lem}
\textit{{Proof: }}
    According to the Gershgorin disk theorem \cite{Johnson}, at the synchronous state, there exists an index $i\in\mathcal{V}_H$ such that
              \begingroup
\setlength{\abovedisplayskip}{3pt}
\setlength{\belowdisplayskip}{3pt}
\begin{align*}
    |\lambda_{\min}(\bm V_{vv}) - \bm V_{vv}(i,i)| \le \sum_{j\neq i} |\bm V_{vv}(i,j)|.
\end{align*}
\endgroup
By Assumption \ref{assump:system} (4) and (6), we obtain a lower bound for $\lambda_{\min}(\bm V_{vv})$,
\begingroup
\begin{small}\setlength{\abovedisplayskip}{3pt}
\setlength{\belowdisplayskip}{3pt}
        \begin{align}
            \lambda_{\min}(\bm V_{vv})&\geq\bm V_{vv}(i,i)-\sum_{j\neq i} |\bm V_{vv}(i,j)|\nonumber=\frac{Q_i}{(v_i^*)^2}-\sum_{j\in\mathcal{V}_H} b_{ij}\cos\theta_{ij}^*\ge\frac{C_2}{\lambda_{\min}\left(\bm V_{\theta\theta}\right)}> 0,\label{Vvv:eig}
        \end{align}
    \end{small}
    \endgroup
where $\lambda_{\min}(\bm V_{\theta\theta})>0$ follows from Lemma \ref{lem:Vtheta}. Hence, $\lambda_{\min}(\bm V_{vv})>0$ and \(\bm V_{vv}\succ 0\).
\hfill $\square$
\par
\begin{lem}\label{lem:G}
Under Assumption~\ref{assump:system}, the matrix
\begingroup
\setlength{\abovedisplayskip}{3pt}
\setlength{\belowdisplayskip}{3pt}
        \begin{align*}
            \bm G=\begin{bmatrix}
                \bm V_{\theta\theta}&\bm V_{\theta v} \\
                \bm V_{\theta v}^\top&\bm V_{v v}
            \end{bmatrix}
        \end{align*}
\endgroup
is positive definite.
\end{lem}
\par
\textit{{Proof: }}
    By the Schur complement, the matrix $\bm G$ is positive definite if and only if
  \begingroup
\setlength{\abovedisplayskip}{3pt}
\setlength{\belowdisplayskip}{3pt}
        \begin{align*}
            \bm V_{v v}\succ 0,\quad\, 
            \bm V_{\theta\theta}-\bm V_{\theta v}\bm V_{v v}^{-1}\bm V_{\theta v}^\top\succ 0.
    \end{align*}
    \endgroup
    The first condition has already been established in Lemma \ref{lem:Vv}. We now verify the second condition. $\forall \bm y\in\mathbb{R}^{n_H-1}$,
  \begingroup
\setlength{\abovedisplayskip}{3pt}
\setlength{\belowdisplayskip}{3pt}
        \begin{align*}
            \bm y^\top\bm V_{\theta v}\bm V_{v v}^{-1}\bm V_{\theta v}^\top\bm y&\leq \lambda_{\max}\left(\bm V_{v v}^{-1}\right)\sigma_{\max}(\bm V_{\theta v})^2\|\bm y\|^2\\
            &\leq \|\bm V_{\theta v}\|_2^2/{\lambda_{\min}\left(\bm V_{v v}\right)}\|\bm y\|^2\\
            &\leq \|\bm V_{\theta v}\|_F^2/\lambda_{\min}(\bm V_{vv})\|\bm y\|^2.
    \end{align*}
    \endgroup
     The second inequality holds due follows from $\|\bm V_{\theta v}\|_2\le \|\bm V_{\theta v}\|_F$. In the third inequality, using $\|\bm V_{\theta v}\|_2\leq \|\bm V_{\theta v}\|_F$ and Assumption \ref{assump:system} (6), we have
     \(
        \|\bm V_{\theta v}\|_{F}^2\le C_1.
\)
\par
Note that Lemma \ref{lem:Vtheta} and Lemma \ref{lem:Vv} guarantee $\lambda_{\min}(\bm V_{\theta\theta}) > 0$ and $\lambda_{\min}(\bm V_{vv}) > 0$. Using the result \(\lambda_{\min}(\bm V_{vv})\ge{C_2}/{\lambda_{\min}\left(\bm V_{\theta\theta}\right)}\) in \eqref{Vvv:eig} and the definition of \(C_2\), we have
\begingroup
\begin{small}
\setlength{\abovedisplayskip}{5pt}
\setlength{\belowdisplayskip}{5pt}
\begin{equation*}
    \begin{aligned}
        \bm y^\top(\bm V_{\theta\theta}-\bm V_{\theta v}\bm V_{v v}^{-1}\bm V_{\theta v}^\top)\bm y &\ge \Big(\lambda_{\min}(\bm V_{\theta\theta}) - \frac{C_1}{\lambda_{\min}(\bm V_{vv})}\Big)\|\bm y\|^2\geq\frac{C_2-C_1}{\lambda_{\min}\left(\bm V_{vv}\right)}\|\bm y\|^2> 0.
    \end{aligned}
\end{equation*}
\end{small}
\endgroup
  The last inequality follows directly from Assumption \ref{assump:system} (6). Hence, $\bm V_{\theta\theta}-\bm V_{\theta v}\bm V_{v v}^{-1}\bm V_{\theta v}^\top\succ 0$.
    \par
Combining both conditions, we conclude that $\bm G\succ 0$.
\hfill $\square$
\par
{
The above lemmas provide the Hessian conditions for the Lyapunov-based analysis. We now use the candidate Lyapunov function and LaSalle's invariance principle to establish local asymptotic stability in the reduced coordinates.
}
\begin{thm}{\em Local stability of the synchronous-state trajectory.}\label{thm:stability}
Consider the system \eqref{eq_nonlinear} under Assumption~\ref{assump:system}. 
For any $\epsilon >0$, there exists a $\delta(\epsilon)>0$ such that, if $\|\widetilde{\bm x}(0)\|<\delta$, then for all $t>0$, 
\[
\|\widetilde{\bm x}(t;0,\widetilde{\bm x}(0))\|<\epsilon,
\]
and
\[
\lim_{t\to\infty} \widetilde{\bm x}(t;0,\widetilde{\bm x}(0)) = \bm 0.
\]
Thus, the origin of the reduced deviation system is locally asymptotically stable, and the synchronous-state trajectory $\bm x^s(t)$ is locally asymptotically stable modulo uniform phase shifts.
\end{thm}
    \textit{{Proof: }}
We first establish the following properties of the Lyapunov function on a sufficiently small neighborhood $\Omega_r$ of the origin:
\begingroup
\setlength{\abovedisplayskip}{3pt}
\setlength{\belowdisplayskip}{3pt}
    \begin{align*}
&V(\bm 0) = 0, ~
 V(\widetilde{\bm x})>0,~~\forall ~ \widetilde{\bm x} \in \Omega_r \backslash \{\bm 0\},\\
&\nabla V(\bm0)=\bm 0,~~ \nabla^2 V(\bm0)\succ\bm0,
\\ &\dot V(\widetilde{\bm x})\leq0, ~~ \forall\,\widetilde{\bm x}\in\Omega_r. 
\end{align*} 
Then we construct a compact positively invariant subset $\Omega_c\subset\Omega_r$ and show that the largest invariant set contained in $\{\widetilde{\bm x} \in \Omega_c \mid \dot V(\widetilde{\bm x}) = 0 \}$ is $\{\bm 0\}$. These results yield the local asymptotic stability of the origin and the convergence of trajectories to it by Lyapunov's method and LaSalle's invariance principle.
\par
Substitute the formulas of $U(\bm\theta,\bm v)$ into $V(\widetilde{\bm x})$, we have
\begingroup
\setlength{\abovedisplayskip}{3pt}
\setlength{\belowdisplayskip}{3pt}
    \begin{align*}
            V(\widetilde{\bm x})=&\sum_{i\in\mathcal{V}_H} \mathopen{}\Biggl[\frac{1}{2}m_i\widetilde{\omega}_i^2\!+\!k_i\widetilde{v_i}\!-\!Q_i\ln \frac{v_i}{v_i^*}-\frac{1}{2}b_{ii}(v_i^2-(v_i^*)^2)\!-\!(P_i\!-\!d_i\omega_i^*)\widetilde{\theta_i}\\
            & \mathopen{}-\mkern-10mu\sum_{j\in\mathcal{N}_i^{\text{high}}}\mkern-10mu b_{ij}\Big(v_iv_j\cos(\theta_{ij}^*\!+\!\widetilde{\theta}_i\!-\!\widetilde{\theta}_j)
            \!-\!v_i^*v_j^*\cos\theta_{ij}^*\Big)\Biggr]\mathclose{},
    \end{align*}
    \endgroup
    which satisfies $ V(\bm 0)=0$.  The gradient vector of the coupling term
\begingroup
\setlength{\abovedisplayskip}{3pt}
\setlength{\belowdisplayskip}{3pt}
    \begin{align*}
         W=\sum_{i\in\mathcal{V}_H} \!\sum_{j\in\mathcal{N}_i^{\text{high}}} b_{ij}v_iv_j\cos\theta_{ij},
    \end{align*}
    \endgroup
    is differentiated at $\theta_{ij}=\theta_{ij}^*+\widetilde{\theta}_{i}-\widetilde{\theta}_j$, $v_i=v_i^*+\widetilde{v}_i$ for all $i\in\mathcal{V}_H$, 
    \begingroup
\setlength{\abovedisplayskip}{3pt}
\setlength{\belowdisplayskip}{3pt}
    \begin{align*}
           \frac{\partial W}{\partial \widetilde{\theta}_i}&=\frac{\partial}{\partial \widetilde{\theta}_i}\mathopen{}\biggl[\sum_{j\in\mathcal{N}_i^{\text{high}}} \mkern-3mu b_{ij}v_iv_j\cos\theta_{ij} \!+\!\mkern-20mu \sum_{\substack{k\in\mathcal{V}_H\\ ~~~i\in\mathcal{N}_k^{\text{high}}}} \mkern-3mu b_{ki}v_iv_k\cos\theta_{ki}\biggr]\mathclose{}\\
           &=\frac{\partial}{\partial \widetilde{\theta}_i}\mathopen{}\biggl[\sum_{j\in\mathcal{N}_i^{\text{high}}} \mkern-3mu b_{ij}v_iv_j\cos\theta_{ij}\!+\!\mkern-10mu\sum_{k\in\mathcal{N}_i^{\text{low}}} \mkern-3mu b_{ik}v_iv_k\cos\theta_{ik}\biggr]\mathclose{}\\
           &=-\sum_{(j,i)\in\mathcal{E}} b_{ij}v_iv_j\sin\theta_{ij},\\
        \frac{\partial W}{\partial \widetilde{v}_i}&=\sum_{(j,i)\in\mathcal{E}}b_{ij}v_j\cos\theta_{ij}.
    \end{align*}
    \endgroup
    Using these derivatives of the coupling term, we calculate the elements of the gradient vector $\nabla V(\widetilde{\bm x})$ of the Lyapunov function, 
\begingroup
\setlength{\abovedisplayskip}{3pt}
\setlength{\belowdisplayskip}{3pt}
    \begin{align*}
            \frac{\partial V}{\partial\widetilde{ \omega}_i}&=m_i\widetilde{\omega}_i,\\
              \frac{\partial V}{\partial\widetilde{ \theta}_i}&=-P_i+d_i\omega^*+\sum_{(j,i)\in\mathcal{E}} b_{ij}v_iv_j\sin\theta_{ij},\\
            \frac{\partial V}{\partial\widetilde{v}_i}&=k_i-\frac{Q_i}{v_i}-\sum_{(j,i)\in\mathcal{E}}b_{ij}v_j\cos\theta_{ij}-b_{ii}v_i,
        \end{align*}
    \endgroup
    {Substitute \eqref{synstate} into the above gradient vector, we have $\nabla V(\bm 0)=\bm 0$.} Subsequently, the elements of the Hessian matrix of the Lyapunov function are derived from the above formulas, for all $ i,j\in\mathcal{V}_H$ and $i\neq j$,
   \begingroup
\setlength{\abovedisplayskip}{3pt}
\setlength{\belowdisplayskip}{3pt}
    \begin{align}\label{Hessian}
            \frac{\partial^2 V}{\partial\widetilde{\omega}_i^2}\Big|_{\bm 0}&=m_i,\nonumber\\
            \frac{\partial^2 V}{\partial\widetilde{\theta}_i^2}\Big|_{\bm 0}&=\sum_{(j,i)\in\mathcal{E}} b_{ij}v_i^*v_j^*\cos\theta_{ij}^*,\nonumber\\
            \frac{\partial^2 V}{\partial\widetilde{\theta}_i\partial\widetilde{\theta}_j}\Big|_{\bm 0}&=-b_{ij}v_i^*v_j^*\cos\theta_{ij}^*,\nonumber\\
            \frac{\partial^2 V}{\partial\widetilde{ v}_i^2}\Big|_{\bm 0}&=\frac{Q_i}{(v_i^*)^2}-b_{ii},\\
           \frac{\partial^2 V}{\partial \widetilde{ v}_i\partial \widetilde{v}_j}\Big|_{\bm 0}&=-b_{ij}\cos\theta_{ij}^*,\nonumber\\
            \frac{\partial^2 V}{\partial \widetilde{\theta}_i\partial \widetilde{v}_i}\Big|_{\bm 0}&=\frac{\partial^2 V}{\partial \widetilde{ v}_i\partial \widetilde{\theta}_i}\Big|_{\bm 0}=\sum_{(j,i)\in\mathcal{E}} b_{ij}v_j^*\sin\theta_{ij}^*,\nonumber\\
            \frac{\partial^2 V}{\partial \widetilde{\theta}_i\partial \widetilde{ v}_j}\Big|_{\bm 0}&=\frac{\partial^2 V}{\partial \widetilde{ v}_j\partial \widetilde{\theta}_i}\Big|_{\bm 0}=b_{ij}v_i^*\sin\theta_{ij}^*\nonumber.
        \end{align}
        \endgroup
    The Hessian matrix of $V$ at the synchronous state is
 \begingroup
\setlength{\abovedisplayskip}{3pt}
\setlength{\belowdisplayskip}{3pt}
    \begin{align*}
            \nabla^2 V(\bm 0)=\begin{bmatrix}
                \bm V_{\theta\theta}&0&\bm V_{\theta v}\\
                0&\bm V_{\omega\omega}&0\\
                \bm V_{\theta v}^\top&0&\bm V_{v v}\\
            \end{bmatrix}\in\mathbb{R}^{(3n_H-1)\times (3n_H-1)},
    \end{align*}
    \endgroup
    where $\bm V_{\theta\theta}$, $\bm V_{\omega\omega}$, $\bm V_{vv}$, and $\bm V_{\theta v}$ are defined in Assumption \ref{assump:system} (4), whose entries correspond to the second-order partial derivatives of $V(\widetilde{\bm x})$ at the synchronous state. \\
    By Assumption \ref{assump:system} (4) and Lemma \ref{lem:G}, we have
  \begingroup
\setlength{\abovedisplayskip}{3pt}
\setlength{\belowdisplayskip}{3pt}
    \begin{align*}
            \bm G=\begin{bmatrix}
                \bm V_{\theta\theta}&\bm V_{\theta v} \\
                \bm V_{\theta v}^\top&\bm V_{v v}.
            \end{bmatrix}\succ 0,
    \end{align*}
    \endgroup
Since $\nabla V(\bm 0)=\bm 0$ and $\nabla^2 V(\bm 0)\succ\bm 0$, the second-order sufficient condition for a strict local minimum implies that $\bm 0$ is a strict local minimum of $V$.
\\
{Next, we calculate the time derivative of the Lyapunov function as follows.}
\begingroup
\setlength{\parskip}{3pt} 
\setlength{\abovedisplayskip}{3pt}
\setlength{\belowdisplayskip}{3pt}\begin{align}
        \dot{V}(\widetilde{\bm x})&=\sum_{i\in\mathcal{V}_H}\Biggl[m_i\dot{\omega_i}\widetilde{\omega}_i\!+\!\dot{v}_i\Bigl(k_i\!-\!\frac{Q_i}{v_i}\mkern-3mu-\mkern-3mu \sum_{\mathclap{(j,i)\in\mathcal E}} \mkern-1mu b_{ij}v_j\cos\theta_{ij}\!-\!b_{ii}v_i\Bigr)\Biggr]+\!\sum_{i\in\mathcal{V}_H}(\widetilde{\omega}_i\!-\!\widetilde{\omega}_g)\Biggl[\!-(P_i\!-\!d_i\omega_i^*)\!+\!\sum_{(j,i)\in\mathcal{E}} b_{ij}v_iv_j\sin\theta_{ij}\Biggr]\notag\\
        &=\sum_{i\in\mathcal{V}_H}\Biggl[-d_i\widetilde{\omega}_i^2-\frac{\dot{v_i}}{v_i}(-k_iv_i+Q_i+\mkern-10mu\sum_{j\in\mathcal{V}_H} b_{ij}v_iv_j\cos\theta_{ij})\Biggr]\notag\\
        &=-\sum_{i\in\mathcal{V}_H}\Bigl(d_i\widetilde{\omega}_i^2+\frac{\tau_i}{v_i}(\dot{v_i})^2\Bigr)\notag.
 \end{align}
\endgroup
       {
By Assumption \ref{assump:system} (3), we have $\dot{V}(\widetilde{\bm x})\leq 0$, while $\dot{V}(\widetilde{\bm x})=0$ implies $\widetilde{\omega}_i=0$ and $\dot v_i=0$.}\\
Using the Taylor approximation of the Lyapunov function with the integral remainder term, we obtain
\begingroup
\setlength{\abovedisplayskip}{3pt}
\setlength{\belowdisplayskip}{3pt}
\begin{align*}
    V(\widetilde{\bm x})=V(\bm 0)+\nabla V(\bm 0)^\top\widetilde{\bm x}+\int_{0}^1 (1-s)\widetilde{\bm x}^\top\nabla^2 V(s\widetilde{\bm x})\widetilde{\bm x}\text{d}s.
\end{align*}
\endgroup
      Since $V(\bm 0)=0$ and $\nabla V(\bm 0)=0$, this reduces to
\begingroup
\setlength{\abovedisplayskip}{3pt}
\setlength{\belowdisplayskip}{3pt}
\begin{align*}
     V(\widetilde{\bm x})=\frac{1}{2}\widetilde{\bm x}^\top\nabla^2 V(\bm 0)\widetilde{\bm x}+R(\widetilde{\bm x}),
\end{align*}
\endgroup
    where  
  \begingroup
\setlength{\abovedisplayskip}{3pt}
\setlength{\belowdisplayskip}{5pt}
\begin{align*}
      R(\widetilde{\bm x})=\int_{0}^1 (1-s)\widetilde{\bm x}^\top\left(\nabla^2 V(s\widetilde{\bm x})-\nabla^2 V(\bm 0)\right)\widetilde{\bm x}\text{d}s.
\end{align*}
\endgroup
     Under Assumption \ref{assump:system} (5), $\nabla^2 V$ is $L$-Lipschitz continuous. Hence,
\begingroup
\setlength{\abovedisplayskip}{3pt}
\setlength{\belowdisplayskip}{3pt}
\begin{align*}
            |R(\widetilde{\bm x})|
        &\leq L\|\widetilde{\bm x}\|^2\left|\int_{0}^1 (1-s)\|s\widetilde{\bm x}\|\text{d}s\right|\leq L\|\widetilde{\bm x}\|^3\left|\int_{0}^1 s(1-s)\text{d}s\right|
        =\frac{1}{6}L\|\widetilde{\bm x}\|^3.
\end{align*}
\endgroup
    So the Lyapunov function has a lower bound
  \begingroup
\setlength{\abovedisplayskip}{3pt}
\setlength{\belowdisplayskip}{3pt}
\begin{align*}
           V(\widetilde{\bm x})&\geq \frac{1}{2}\widetilde{\bm x}^\top\nabla^2 V(\bm 0)\widetilde{\bm x}-\frac{1}{6}L\|\widetilde{\bm x}\|^3\ge\left[\frac{1}{2}\lambda_{\min}(\nabla^2 V(\bm 0))-\frac{1}{6}L\|\widetilde{\bm x}\|\right]\|\widetilde{\bm x}\|^2.
\end{align*}
\endgroup
    In addition, an upper bound of the Lyapunov function is 
  \begingroup
\setlength{\abovedisplayskip}{3pt}
\setlength{\belowdisplayskip}{3pt}
\begin{align*}
            V(\widetilde{\bm x})&\leq \left[\frac{1}{2}\lambda_{\max}(\nabla^2 V(\bm 0))+\frac{1}{6}L\|\widetilde{\bm x}\|\right]\|\widetilde{\bm x}\|^2.
\end{align*}
\endgroup
 According to Assumption \ref{assump:system} (7), we have 
\begin{equation}\label{c10}
    \begin{aligned}
        c_2>c_1 &= \frac12 \lambda_{\min}(\nabla^2 V(\bm 0)) - \frac{L r}{6}>\frac12 \lambda_{\min}(\nabla^2 V(\bm 0))-\frac{3 \lambda_{\min}(\nabla^2 V(\bm 0))}{6}=0.
    \end{aligned}
\end{equation}
Thus, for all {$\widetilde{\bm x} \in \Omega_r$}, 
\begin{equation*}
    \scalebox{0.95}{$
     \begin{aligned}
        \frac{1}{2}\lambda_{\min}(\nabla^2 V(\bm 0))-\frac{1}{6}L\|\widetilde{\bm x}\|&\ge  \frac12 \lambda_{\min}(\nabla^2 V(\bm 0)) - \frac{L r}{6}=c_1,\\
    \frac{1}{2}\lambda_{\max}(\nabla^2 V(\bm 0))+\frac{1}{6}L\|\widetilde{\bm x}\|&\le \frac12 \lambda_{\max}(\nabla^2 V(\bm 0)) + \frac{L r}{6} =c_2.
    \end{aligned}
    $}
\end{equation*}
Therefore, for every $\widetilde{\bm x}\in\Omega_r$,
\begin{equation}\label{bound}
    0\le c_1 \|\widetilde{\bm x}\|^2
    \le V(\widetilde{\bm x})
    \le c_2 \|\widetilde{\bm x}\|^2.
\end{equation}
Since $c_1>0$ and $V(\bm 0)=0$, it follows that
\begin{equation*}
    V(\widetilde{\bm x})>0,
    ~~
    \forall\,\widetilde{\bm x}
    \in\Omega_r\backslash\{\bm 0\}.
\end{equation*}
Thus, $V$ is positive definite on $\Omega_r$.
{
   For any $\epsilon>0$, recall from Assumption \ref{assump:system} (7) that  $\delta(\epsilon) = \sqrt{{c_1}/{2c_2}}\min\{\epsilon, r\}$, and define
$\Omega_c=\{ \widetilde{\bm x} \in\Omega_r \mid V(\widetilde{\bm x}) \le c_2 \delta^2(\epsilon) \}$. 
Since $V(\widetilde{\bm x})$ is continuous and $\Omega_r$ is compact, $\Omega_c$ is compact.
In addition, on $\partial\Omega_r$, $V(\widetilde{\bm x})\ge c_1r^2> c_2\delta^2(\epsilon)$, which implies $\Omega_c\cap\partial\Omega_r=\emptyset$.
Because $\dot{V}(\widetilde{\bm x})\le 0$, every trajectory starting in $\Omega_c$ satisfies, for all $t\in (0,\infty)$, $V(\widetilde{\bm x}(t))\le V(\widetilde{\bm x}(0))\le c_2 \delta^2(\epsilon)$. Hence, any such trajectory cannot reach $\partial\Omega_r$ and remains in $\Omega_c$ for all $t\in(0,\infty)$.
Therefore, $\Omega_c$ is a compact positively invariant subset of $\Omega_r$.
\par
Let $\Omega_E=\{\widetilde{\bm x} \in \Omega_c \mid \dot V(\widetilde{\bm x}) = 0 \}$ and let $\Omega^{\text{lar}}$ be the largest invariant set contained in $\Omega_E$.
Since $\Omega^{\mathrm{lar}}\subset\Omega_E$, every trajectory in $\Omega^{\mathrm{lar}}$ satisfies $\widetilde{\omega}_i=0$ and $\dot v_i=0$ for all time.  
For any such trajectory, the invariance property requires $\dot{\widetilde{\omega}}_i=0$.
Together with $\dot v_i=0$, every trajectory in $\Omega^{\mathrm{lar}}$ satisfies the definition of synchronous state \eqref{synstate}. 
By the local uniqueness assumption, there is only one such synchronous state in $\Omega_r$, which is represented by
$\widetilde{\bm x}=\bm 0$. Hence, $\Omega^{\mathrm{lar}}=\{\bm 0\}$.
By LaSalle's invariance principle \cite{Vidyasagar}, every trajectory starting in
$\Omega_c$ converges to $\Omega^{\mathrm{lar}}$. Therefore,
\begin{equation}\label{lasalle}
  \lim_{t\to\infty}\widetilde{\bm x}(t)=\bm 0.
\end{equation}
Assume that $\|\widetilde{\bm x}(0)\|<\delta(\epsilon)$.
Then \begingroup
\setlength{\abovedisplayskip}{3pt}
\setlength{\belowdisplayskip}{5pt}
\begin{align*}
    V(\widetilde{\bm x}(0)) \le c_2\|\widetilde{\bm x}(0)\|^2< c_2\delta^2(\epsilon).
\end{align*}
\endgroup
Because the Lyapunov function is non-increasing along trajectories, for all $t\in(0,\infty)$,
\begingroup
\setlength{\abovedisplayskip}{3pt}
\setlength{\belowdisplayskip}{5pt}
\begin{equation}\label{compare}
V(\widetilde{\bm x}(t)) \le V(\widetilde{\bm x}(0)) < c_2 \delta^2(\epsilon)< c_1\min\{\epsilon^2,r^2\}.
\end{equation}\endgroup} 
Therefore, for all $t\in(0,\infty)$, $\widetilde{\bm x}(t)\in\Omega_c$. 
We have
\begingroup
\setlength{\abovedisplayskip}{3pt}
\setlength{\belowdisplayskip}{5pt}
\begin{equation}
    c_1\|\widetilde{\bm x}(t)\|^2\le V(\widetilde{\bm x}(t)) < c_1\min\{\epsilon^2,r^2\}.
\end{equation}\endgroup
Because $c_1>0$ and $\widetilde{\bm x}(t)$ is continuous in time, then the condition $\|\widetilde{\bm x}(0)\| < \delta(\epsilon)$ implies that, for all $t\in(0,+\infty)$,
\begin{equation}
     \|\widetilde{\bm x}(t)\|<\min\{\epsilon,r\}\le \epsilon.
\end{equation}
The above analysis proves the local stability of the origin of the deviation system. Together with the convergence result in \eqref{lasalle}, it follows that the origin of the deviation system is locally asymptotically stable. Consequently, the relative phase angles, frequencies, and voltages of the original system converge to their corresponding synchronous states. Thus, the synchronous-state trajectory is locally asymptotically stable with respect to the deviation variables.
\hfill $\square$
\begin{rem}
    {
    Although Theorem \ref{thm:stability} establishes local asymptotic stability of the open-loop system, it does not address the frequency and voltage regulation. The ESS controllers introduced in the next section are designed to regulate the frequency and voltage through active and reactive power inputs.}
\end{rem}
\section{Distributed control structure of the power system}\label{sec:distributed_control}
In this multilevel distribution network, Levels 0--2 provide aggregated demand information for power scheduling. At Level 0, the controller of each house computes its power demand from the grid and communicates to the scheduling controller of Level 1 for the next five minutes. No ESS-based controllers are modeled for Levels 1 and 2. This section focuses on the dynamic models and ESS-based distributed control of the Level buses at Levels 3--5, where active- and reactive-power ESS inputs are used for frequency and voltage regulation.
\begin{defn}{\em Two-step control framework.}
  The control framework operates in two steps during each five-minute horizon:
\begin{itemize}
\item
\textbf{Scheduling for the next horizon.}
    Before each horizon, the scheduled power inflows and outflows are computed based on the predicted power supply and demand. These scheduled values are then kept fixed for that horizon.
\item
\textbf{Feedback regulation within a short time horizon.}
Local ESS controllers use state feedback to reduce frequency and voltage variances caused by fluctuations of power loads, ESS power supplies, and local power sources.
\end{itemize}
\end{defn}
The control objectives for Levels 3--5 are as follows:
\par
\textbf{Power balancing.}
   Balance the power supply and the predicted power demand by recomputing the scheduled power flows for each five-minute horizon.
\par
\textbf{Frequency and voltage regulation.}
  Reduce the variances of the voltages and of the frequencies through ESS feedback.
\par
This paper focuses on frequency and voltage regulation during a short horizon of about five minutes.
Control laws are synthesized using one of the following two methods:
\begin{itemize}
\item \textbf{Distributed control.}
Each local controller uses only
information of the local power system component and
communicates the scheduled power inflows and outflows to the controller of the next higher level every five minutes.
\item \textbf{Centralized control.}
A single controller controls the
entire distribution network of Levels 3, 4, and 5 based on
the information from all power systems of these levels.
\end{itemize}
The distributed control method is adopted for Levels 3, 4, and 5 because it reduces the requirements for information collection
and for computation. The centralized control method is used for a comparison with the distributed control.
\par
In the combined system \eqref{combined_eq}, inter-level power flows depend on phase angles and voltage magnitudes, thereby coupling the dynamics of buses. For local controller synthesis, these flows are approximated by their scheduled quasi-static values.
\begin{defn}{\em Isolated power system at Level $a$.} 
For each \(a\in\mathcal L_H\) and \(i\in\mathcal L_a\), let \(\bar P_{i,j}\) and \(\bar Q_{i,j}\) denote the scheduled quasi-static active- and reactive-power flows between adjacent levels. The isolated model of bus \(i\in\mathcal V_H\), where these flows are treated as exogenous signals, is given by
\begingroup 
\setlength{\parskip}{3pt}
\setlength{\abovedisplayskip}{3pt}
\setlength{\belowdisplayskip}{3pt} \begin{align}\label{isolated} \frac{\text{d}\theta_i(t)}{\text{d}t}&=\omega_i(t),\nonumber\\ m_i\frac{\text{d}\omega_i(t)}{\text{d}t}&=-d_i\omega_i(t)+{P}_i-\sum_{j\in\mathcal{V}_H}\bar{P}_{i,j}+u_{i,\omega}(t),\nonumber\\ \tau_i\frac{\text{d}v_i(t)}{\text{d}t}&=-k_iv_i(t)+Q_i-\sum_{j\in\mathcal{V}_H}\bar{Q}_{i,j}+u_{i,v}(t),\nonumber\\ 
     P_i&=P_{sp,i}-P_{ld,i},\\ Q_i&=Q_{sp,i}-Q_{ld,i},\nonumber\\ P_{sp,i}&=P_{sol,i}+P_{wind,i}+P_{bm,i},\nonumber\\ 
     Q_{sp,i}&=Q_{sol,i}+Q_{wind,i}+Q_{bm,i}.\nonumber 
    \end{align}
     \endgroup
\end{defn}
The approach to frequency and voltage regulation consists of the following steps:
\par
    \begin{enumerate}
         \item Use the nonlinear dynamic system \eqref{combined_eq} and its synchronous state as the reference for feedback control.
         \item Derive reduced-coordinate linear models from the
    isolated system \eqref{isolated}.
        \item Design local LQR feedback laws for the ESS input power at each of Levels 3, 4, and 5.
        \item Apply the resulting distributed controller to the
    combined system and establish conditions for closed-loop
    stability in Section~\ref{sec:closed-loop}.
        \item Compare the quadratic-cost performance of the
    distributed and centralized controllers for the linearized
    model in the performance analysis.
    \end{enumerate}
\begin{defn}{\em State-space representation of isolated power system \eqref{isolated} for Levels 3, 4, and 5.} 
 The isolated power system \eqref{isolated} can be represented in state-space form as follows:
\begin{equation}\label{linear}
    \begin{aligned}
       & \frac{\text{d}\widetilde{\bm x}(t)}{\text{d}t}\!=\!\bm{A}\widetilde{\bm x}+\bm{B}\bm{u},~~ \widetilde{\bm x}(0)\!=\!\mathrm{col}(\widetilde{\bm x}_g(0),(\widetilde{\bm x}_i(0))_{i\in\mathcal{I}}),
    \end{aligned}
\end{equation}
where $\bm u=\text{col}(\bm u_i)\in\mathbb{R}^{2n_H}$ is the vector of ESS control inputs with $\bm u_i=(u_{i,\omega},u_{i,v})$. The input matrix $\bm{B}$ has a block-diagonal structure. The matrices defining the local dynamics and the reference-frequency coupling are given below.
 \begin{subequations}
      \begin{align}
        \bm A_{gg} &=\begin{bmatrix}
            -d_g/m_g& 0\\
           0&-k_g/\tau_g
        \end{bmatrix},~~  
        \bm B_g=\begin{bmatrix}
    1/m_g&0\\
    0&1/\tau_g
  \end{bmatrix},\nonumber\\
        \text{For}&~~i\in\mathcal{I},\nonumber\\
          \bm A_{ii}&=\begin{bmatrix}
            0 &1&0\\
            0 &-d_i/m_i& 0\\
           0&0&-k_i/\tau_i
        \end{bmatrix},~~
        \label{statematrix}\\
        \bm B_i&=\begin{bmatrix}
            0&0\\
            1/m_i &0\\
            0&1/\tau_i
        \end{bmatrix},
       \label{inputmatrix}\\ 
       \bm E_{ig}&=\begin{bmatrix}
         -1 &0\\
         0&0\\
         0&0
        \end{bmatrix},~~
        \bm E_{gg}=\begin{bmatrix}
         0&0\\
         0&0
        \end{bmatrix},\nonumber\\
        \bm E_{ij}&=\bm 0, \text{otherwise},\notag\\
        \bm E&=(\bm E_{ij})\in\mathbb{R}^{(3n_H-1)\times(3n_H-1)},\nonumber\\
        \bm A_0&=\mathrm{diag}(\bm A_{gg},(\bm A_{ii})_{i\in\mathcal{I}})\in\mathbb{R}^{(3n_H-1)\times(3n_H-1)},\nonumber\\
         \bm A&=(\bm A_0+\bm E)\in\mathbb{R}^{(3n_H-1)\times(3n_H-1)},\nonumber\\
        \bm B&=\mathrm{diag}(\bm B_g,(\bm B_i)_{i\in\mathcal{I}})\in\mathbb{R}^{(3n_H-1)\times 2n_H}.\nonumber
\end{align}
\end{subequations}
{The matrix $\bm A_0$ contains the local diagonal blocks,
while $\bm E$ describes the coupling induced by the
reduced phase-angle coordinates.}
\end{defn}
For the reference node \(g\), the reduced state is \(\widetilde{\bm x}_g=(\widetilde\omega_g,\widetilde v_g)^\top\), and for each node \(i\ne g\), it has the state \(\widetilde{\bm x}_i=(\widetilde\theta_i,\widetilde\omega_i,\widetilde v_i)^\top\). The output matrix is defined as
\[
\bm C_g=\operatorname{diag}(C_g^\omega,C_g^v)\in\mathbb R^{2\times2},
\]
and, for \(i\ne g\),
\begin{equation*}
    \begin{aligned}
        \bm C_i&=\mathrm{diag}(C_i^\theta,C_{i}^\omega,C_{i}^v),
    \end{aligned}
\end{equation*}
where all diagonal entries are positive. The global output matrix is
\begin{equation}\label{outputmatrix}
\bm C=\operatorname{diag}\bigl(\bm C_g,(\bm C_i)_{i\in\mathcal{I}}\bigr)
\in\mathbb R^{(3n_H-1)\times(3n_H-1)}.
\end{equation}

\begin{lem}{\em Controllability of the isolated power system. }\label{controllable}
    Under the given state matrix \eqref{statematrix} and input matrix \eqref{inputmatrix}, for $i\in\mathcal{V}_H$, the pair $(\bm A_{ii},\bm B_i)$ is controllable.
\end{lem}

\textit{{Proof: }} 
       For each $i\in\mathcal{I}$, by the left-eigenvector characterization \cite[Th.~3.13]{Trentelman}, $(\bm A_{ii},\bm B_i)$ is controllable if and only if
\[
\forall\lambda\in\text{spec}(\bm A_{ii}),\quad
\bm y^\top\bm A_{ii}=\lambda \bm y^\top,\, \bm y^\top\bm B_i=\bm 0\, \Rightarrow \bm y=0.
\]
Define a vector $\bm y=\text{col}(\bm y_i)\in \mathbb{R}^{3}$ with $\bm y_i=(y_{i}^\theta,y_{i}^\omega,y_{i}^v)$. From $\bm y^\top \bm B_i=0$, we obtain 
\[
y_{i}^\omega/m_i=0,\quad y_{i}^v/\tau_i=0.
\]
Since $m_i,\tau_i>0$, it follows that $y_{i}^\omega=y_{i}^v=0$. Substituting these formulas into $\bm y^\top \bm A_{ii} = \lambda \bm y^\top$ gives
\[
 y_{i}^\theta = 0.
\]
Hence, $\bm y = \bm 0$, implying that $(\bm A_{ii}, \bm B_i)$ is controllable.
For the reference node $g$, controllability of the pair $(\bm A_{gg},\bm B_g)$ follows by an analogous argument.
\hfill $\square$
\begin{lem}{\em Observability of the isolated power system. }\label{observable}
   Under the given state matrix \eqref{statematrix} and output matrix \eqref{outputmatrix}, for $i\in\mathcal{V}_H$, the pair $(\bm A_{ii},\bm C_i)$ is observable.
\end{lem}
\textit{{Proof: }} For each $i\in\mathcal{I}$, by the right-eigenvector characterization \cite[Th.~3.13]{Trentelman}, $(\bm A_{ii},\bm C_i)$ is observable if and only if
\begingroup
\[\forall\lambda\in\text{spec}(\bm A_{ii}),\quad
\bm A_{ii}\bm y=\lambda \bm y,\, \bm C_i\bm y=\bm 0\, \Rightarrow\bm y=0.\]
\endgroup
   Define a vector $\bm y=\text{col}(\bm y_i)\in \mathbb{R}^{3}$ with $\bm y_i=(y_{i}^\theta,y_{i}^\omega,y_{i}^v)$. From $ \bm C_i\bm y=0$, we obtain
\begingroup
\[
C_i^\theta y_i^\theta=0,\quad C_i^\omega y_i^\omega=0,\quad C_{i}^v y_{i}^v=0.
\]
\endgroup
Since $C_i^\theta,C_i^\omega, C_{i}^v>0$, it follows that $y_i^\theta=y_{i}^\omega=y_{i}^v=0$.
\par
Hence $\bm y=\bm 0$, implying that $(\bm A_{ii},\bm C_i)$ is observable.
For the reference node $g$, observability of the pair $(\bm A_{gg},\bm C_g)$ follows by an analogous argument.
\hfill $\square$
\par
The controllability and observability results established above, together with the positive-definite weighting matrices of states and inputs, ensure that each local LQR problem determines a unique positive definite solution of the algebraic Riccati equation. The distributed LQR control method of the isolated power systems for Levels 3, 4, and 5 is formulated as follows.
\par
\begin{defn}{\em Distributed LQR control method of the isolated power system for Levels 3, 4, and 5.}
For each node $i\in\mathcal V_H$, a local LQR controller
is designed using the local model
\begin{equation*}
   \frac{\text{d}\widetilde{\bm x}_i(t)}{\text{d}t}={\bm{A}_{ii}\widetilde{\bm x}_i+\bm{B}_i\bm{u}_i}.
\end{equation*}
For the non-reference nodes, the coupling term
$\bm E_{ig}\widetilde{\bm x}_g$ is omitted 
in local LQR synthesis, while it is considered in the closed-loop stability analysis in Section \ref{sec:closed-loop}.
Then the distributed control method is designed as follows, 
\begin{subequations}
\begin{align}
J_{d,i} &= \int_0^\infty
\begin{bmatrix}
\widetilde{\bm x}_i(t) \\
\bm u_i(t)
\end{bmatrix}^\top
\bm Q_{\text{cr},i}
\begin{bmatrix}
\widetilde{\bm x}_i(t) \\
\bm u_i(t)
\end{bmatrix} \mathrm{d}t, \label{cost function}\\
\bm Q_{\text{cr},i} &=
\begin{bmatrix}
\bm Q_{\bm x\bm x,i} & \bm 0 \\
\bm 0 & \bm Q_{\bm u\bm u,i}
\end{bmatrix}\in\mathbb{R}^{(n_{x,i}+2)\times {n_{x,i}+2}}_{pds},\\
n_{x,g}&=2,~~n_{x,i}=3,i\in\mathcal{I},\notag\\
\bm Q_{\bm u\bm u,i} &=\mathrm{diag}(\bm{Q}_{u_{\omega,i}},\bm{Q}_{u_{v,i}})\succ 0, \notag \\
 \bm{C}_i^\top\bm{C}_i&=\bm{Q}_{\bm{x}\bm{x},i},\notag\\
0 &\!=\! \bm Q_i \bm A_{ii} \!+\! \bm A_{ii}^\top \bm Q_i \!+\! \bm Q_{\bm x\bm x,i} \!-\! \bm Q_i \bm B_i \bm Q_{\bm u\bm u,i}^{-1} \bm B_i^\top \bm Q_i, \label{riccati}\\
\bm{Q}_i&=\bm Q_i^\top\succ  0,\notag \\
 \mathbb{C}^{-}&=\{ c \in \mathbb{C} \mid \mathrm{Re}(c) < 0 \},\notag\\
 \mathbb{C}^{-}&\supset\text{spec}(\bm A_{ii} + \bm B_i \bm F_{d,i}), \label{spec} \\
\bm F_{d,i} &= -\bm Q_{\bm u\bm u,i}^{-1} \bm B_i^\top \bm Q_i, \label{feedback}\\
\bm u_{d,i} &= \bm F_{d,i} \widetilde{\bm x}_i 
= \bm F_{d,i}(\widetilde{\theta}_i,\widetilde{\omega}_i,\widetilde{v}_i)^\top, \label{controllaw}\\
J_{d,i}^* &= \min J_{d,i}. \notag
\end{align}
\end{subequations}
\end{defn}
Each local feedback gain minimizes the quadratic cost $J_{d,i}$, where $\bm Q_{\bm x\bm x,i}$ and $\bm Q_{\bm u\bm u,i}$ denote the weights of the state deviation and 
the control inputs. 
\section{Stability analysis of the closed-loop power system}\label{sec:closed-loop}
{The distributed LQR control law is designed for the isolated system \eqref{isolated} and ensures local optimality around the synchronous state defined in Section \ref{sec:open_loop}. However, the isolated model neglects the dynamic couplings between adjacent levels. In this section, the distributed controller is implemented on the combined power system \eqref{combined_eq}, where the effect of these dynamic couplings on the closed-loop dynamics is analyzed. Therefore, we investigate the local stability of the closed-loop system around the synchronous state. For this purpose, a linearized model of the combined power system is first derived.}
\begin{defn}{\em Linearized power system around the synchronous state.}
The nonlinear power system \eqref{combined_eq} is linearized around $\bm{x}^*$,
\begin{equation}\label{original}
    \frac{\text{d}\widetilde{\bm x}(t)}{\text{d}t} =\widetilde{\bm A}\widetilde{\bm x}+\bm{B}\bm{u},
\end{equation}
where $\widetilde{\bm A} = \nabla\bm f(\bm 0)$ is the Jacobian matrix at the synchronous state. The matrix $\widetilde{\bm A}$ can be decomposed as
\begin{equation*}
     \widetilde{\bm A}=\bm A+\bm A^{\bm x}+\widehat{\bm A}\in\mathbb{R}^{(3n_H-1)\times (3n_H-1)} .
\end{equation*}
Define
\begingroup
\setlength{\parskip}{3pt} 
\setlength{\abovedisplayskip}{3pt}
\setlength{\belowdisplayskip}{3pt}
\begin{align*} 
\bm A^{\bm x}_{g}&=\begin{bmatrix}
           0& -\mkern-10mu\sum\limits_{(j,g)\in\mathcal{E}}\mkern-3mu\!\!\!\frac{b_{gj}v_j^*\sin\theta_{gj}^*}{m_g}\\
        0&\frac{\sum\limits_{(j,g)\in\mathcal{E}}\mkern-3mu b_{gj}v_j^*\cos\theta_{gj}^*\!+\!2b_{gg}v_g^*}{\tau_g}
        \end{bmatrix},\\
        \widehat{\bm A}_{ig}&=\begin{bmatrix}
            0&0\\
            0&-\frac{b_{ig}v_i^*\sin\theta_{ig}^*}{m_i}\\
            0&\frac{b_{ig}v_i^*\cos\theta_{ig}^*}{\tau_i}
        \end{bmatrix},~~i\in\mathcal{I},\\
        \widehat{\bm A}_{gj}&=\begin{bmatrix}
            \frac{b_{gj}v_g^*v_j^*\cos\theta_{gj}^*}{m_g}&0&-\frac{b_{gj}v_g^*\sin\theta_{gj}^*}{m_g}\\
            \frac{b_{gj}v_g^*v_j^*\sin\theta_{gj}^*}{\tau_g}&0&\frac{b_{gj}v_g^*\cos\theta_{gj}^*}{\tau_g}
        \end{bmatrix},~~j\in\mathcal{I},\\
        \text{For each}  ~~&i\in\mathcal{I}, \\ 
        \bm A^{\bm x}_{i}&=\begin{bmatrix}
            0 &0&0\\
            -\mkern-10mu\sum\limits_{(j,i)\in\mathcal{E}} \mkern-3mu\!\!\!\frac{ b_{ij}v_i^*v_j^*\cos\theta_{ij}^*}{m_i} &0& -\mkern-10mu \sum\limits_{(j,i)\in\mathcal{E}}\mkern-3mu\!\!\!\frac{b_{ij}v_j^*\sin\theta_{ij}^*}{m_i}\\
            - \mkern-10mu\sum\limits_{(j,i)\in\mathcal{E}}\mkern-3mu\!\!\!\frac{b_{ij}v_i^*v_j^*\sin\theta_{ij}^*}{\tau_i}&0&
\frac{ \sum\limits_{(j,i)\in\mathcal{E}}\mkern-3mu b_{ij}v_j^*\cos\theta_{ij}^*\!+\!2b_{ii}v_i^*}{\tau_i}
        \end{bmatrix}\mkern-3mu,\\
        \widehat{\bm A}_{ij}&=\begin{bmatrix}
            0&0&0\\
            \frac{b_{ij}v_i^*v_j^*\cos\theta_{ij}^*}{m_i}&0&-\frac{b_{ij}v_i^*\sin\theta_{ij}^*}{m_i}\\
            \frac{b_{ij}v_i^*v_j^*\sin\theta_{ij}^*}{\tau_i}&0&\frac{b_{ij}v_i^*\cos\theta_{ij}^*}{\tau_i}
        \end{bmatrix},~~(j,i)\in\mathcal{E},\\
        \widehat{\bm A}_{ij}&=0,~~otherwise,\\ 
        \bm A^{\bm x}&=\mathrm{diag}(\bm A^{\bm x}_{g},(\bm A^{\bm x}_{i})_{i\in\mathcal{I}})\in\mathbb{R}^{(3n_H-1)\times (3n_H-1)},\\
        \widehat{\bm A}&=(\widehat{\bm A}_{ij})\in\mathbb{R}^{(3n_H-1)\times (3n_H-1)}.
    \end{align*}
\endgroup
Here, $\bm A$ is the matrix defined in \eqref{linear}. The diagonal blocks $\bm A^{\bm x}$ represent the state couplings of each Level bus, respectively. The off-diagonal coupling block \(\widehat{\bm A}\) describes the interactions between Level buses.
\end{defn}
The distributed and centralized control laws are,
\begin{align*}
  \bm u_d &= \bm F_d \widetilde{\bm x}, ~~ \bm F_d=\text{diag}(\bm F_{d,i}),\\
   \bm u_c &= \bm F_c \widetilde{\bm x},
\end{align*}
where $\bm F_d$ is the feedback matrix \eqref{feedback} of the distributed control, and \(\bm F_c\) is obtained from a centralized LQR control problem, which will be presented in Section~\ref{sec:performance}.
\par
The stability is analyzed for the following three cases.
\begin{itemize}
    \item Case~1: The linear system \eqref{original} with distributed control.
    \item Case~2: The nonlinear system \eqref{combined_eq} with distributed control.
    \item Case~3: The nonlinear system \eqref{combined_eq} with centralized control.
\end{itemize}
As discussed in \cite{shields:pearson:1976}, there are no simple equivalent conditions for the existence of a stabilizing control law of a time-invariant linear system. Therefore, we establish sufficient conditions for stability using two approaches: (a) controllability conditions for each input component of the distributed linear system, and (b) a bounding argument based on $\mathcal{L}_2$ bounds.
\par
\textbf{Case 1. Linear system with distributed linear control.}
\par
\begin{prop}{\em {Global stability of the linear system with distributed linear control.}} \label{global linear}
Under Assumption~\ref{assump:system} (1)-(3), consider the closed-loop linear system \eqref{original} with the control law $\bm u(t)=\bm F_d\widetilde{\bm x}(t)$, where the matrices $\widetilde{\bm A},\bm B,\bm F_d$ are given.
{
If the distributed feedback matrix $\bm F_d$ satisfies
\begin{equation*}
    \text{spec} (\widetilde{\bm A} + \bm B \bm F_d) \subset\mathbb{C}^{-},
\end{equation*}
then the origin $\widetilde{\bm x}=\bm 0$ of the closed-loop linear system is globally asymptotically stable.}
\end{prop}
\textit{{Proof: }}
    {
The closed-loop linear system is
\begin{equation*}
    \dot{\widetilde{\bm x}}
=
(\widetilde{\bm A}+\bm B\bm F_d)\widetilde{\bm x}.
\end{equation*}
Since $\widetilde{\bm A}+\bm B\bm F_d$ is Hurwitz, 
$\widetilde{\bm x}=\bm 0$ of the
closed-loop linear system is globally asymptotically stable, which follows directly from the stability of linear systems \cite[Th.~5.4.29]{Vidyasagar}.}
\hfill $\square$
\par
There follows another sufficient condition for the existence
of a control law which guarantees the Hurwitz property.
\begin{prop}{\em A sufficient condition of the Hurwitz property under coupling terms.}
        Consider the distributed closed-loop matrix $\widetilde{\bm A}+\bm B\bm F_d=(\bm A_0+\bm B\bm F_d)+\bm A_\Delta$. 
        If \[\|\bm A_\Delta^\top \bm P_d+\bm P_d\bm A_\Delta\|<\lambda_{\min}(\bm Q_d),\] 
        where $\bm A_\Delta= \bm E+\bm A^{\bm x}+\widehat{\bm A}$ represents the coupling terms. Then the distributed closed-loop matrix $\widetilde{\bm A}+\bm B\bm F_d$ is Hurwitz.
\end{prop}
\textit{{Proof: }}
Let $\bm P_d=\text{diag}(\bm Q_i)\succ  0$ and $\bm Q_d=\text{diag}
(\bm Q_{xx,i}+\bm F_{d,i}^\top \bm Q_{uu,i}\bm F_{d,i})\succ 0$, where $\bm Q_i$ is the positive-definite solution of the Riccati equation \eqref{riccati}. Using the Riccati equation \eqref{riccati} and the feedback law \eqref{feedback}, we obtain $(\bm A_0+ \bm B \bm F_d)^\top \bm P_d+\bm P_d(\bm A_0+\bm B\bm F_d)=-\bm Q_d$.
Substituting $\widetilde{\bm A}+\bm B\bm F_d=(\bm A_0+\bm B\bm F_d)+\bm A_\Delta$ yields $(\widetilde{\bm A}+\bm B\bm F_d)^\top \bm P_d+\bm P_d(\widetilde{\bm A}+\bm B\bm F_d)=-\bm Q_d+\bm A_\Delta^\top \bm P_d+\bm P_d\bm A_\Delta$.
For any nonzero vector $\bm y$, it follows that
$\bm y^\top(-\bm Q_d+\bm A_\Delta^\top \bm P_d+\bm P_d\bm A_\Delta)\bm y
\le-\lambda_{\min}(\bm Q_d)\|\bm y\|^2+\|\bm A_\Delta^\top \bm P_d+\bm P_d\bm A_\Delta\|\|\bm y\|^2$.
Hence, if $\|\bm A_\Delta^\top \bm P_d+\bm P_d\bm A_\Delta\|<\lambda_{\min}(\bm Q_d)$, then $(\widetilde{\bm A}+\bm B\bm F_d)^\top \bm P_d+\bm P_d(\widetilde{\bm A}+\bm B\bm F_d)\prec 0$. Therefore, by Lyapunov's theorem, the distributed closed-loop matrix $\widetilde{\bm A}+\bm B\bm F_d$ is Hurwitz.
\hfill $\square$
\par
\textbf{Case 2. Nonlinear system with distributed linear control.}
\par
\begin{prop}{\em {Local stability of the closed-loop nonlinear power system with distributed linear control.}}\label{thm:closed-loop stability}
      Under Assumption~\ref{assump:system} (1)-(3), consider the nonlinear power system \eqref{combined_eq} under the distributed control law \eqref{controllaw}. 
      {If the distributed feedback matrix $\bm F_d$ satisfies
\begin{equation}\label{closed-loop matrix}
 \text{spec}(\widetilde{\bm A}+\bm B\bm F_d)
\subset \mathbb C^- ,  
\end{equation}
then the origin $\widetilde{\bm x}=\bm 0$ of the nonlinear closed-loop system is locally asymptotically stable. Hence, the synchronous-state trajectory $\bm{x}^s$ is locally asymptotically stable. }
\end{prop}
    {
    \textit{{Proof: }}
Under the distributed feedback law, the nonlinear closed-loop system is
written as
\[
\dot{\widetilde{\bm x}}
=
\bm f(\widetilde{\bm x})+\bm B\bm F_d\widetilde{\bm x},
\]
where $\bm f(\bm 0)=\bm 0$. Since $\bm f$ is continuously differentiable
in a neighborhood of the synchronous state and satisfies
$
D\bm f(\bm 0)=\widetilde{\bm A}$, 
the Jacobian matrix of the nonlinear closed-loop system at
$\widetilde{\bm x}=\bm 0$ is
$
\widetilde{\bm A}+\bm B\bm F_d.$
\par
By \eqref{closed-loop matrix}, this Jacobian matrix is Hurwitz. Hence,
Lyapunov's linearization method \cite[Th.~5.5.53]{Vidyasagar} implies
that the equilibrium $\widetilde{\bm x}=\bm 0$ is locally asymptotically
stable. Consequently, the corresponding synchronous-state trajectory
$\bm x^s(t)$ is locally asymptotically stable.
    \hfill $\square$}
\par
\textbf{Case 3. Nonlinear system with centralized control.}
\par
{
\begin{prop}{\em Local stability of the closed-loop nonlinear power system with centralized control.}\label{thm:closed-loop central stability}
      Under Assumption~\ref{assump:system} (1)-(3), consider the nonlinear power system \eqref{combined_eq} under the centralized control law and linearize it at the synchronous state.
   \begin{itemize}
\item[(a)]
The pair $(\widetilde{\bm A},\bm B)$ associated with \eqref{original} is controllable.
\item[(b)]
There exists a centralized feedback matrix $\bm F_c$ such that
$\text{spec} (\widetilde{\bm A} + \bm B \bm F_c) \subset\mathbb{C}^{-}$.
\item[(c)]
If a feedback matrix whose existence is established in part (b) is used,
then the synchronous-state
trajectory $\bm x^s(t)$ is locally asymptotically stable.
\end{itemize}
    \end{prop}}
{
\textit{{Proof: }}
    \emph{Proof of (a).}
    By the left-eigenvector characterization \cite[Th.~3.13]{Trentelman}, $(\widetilde{\bm A},\bm B)$ is controllable if and only if
\begingroup
\setlength{\abovedisplayskip}{3pt}
\setlength{\belowdisplayskip}{5pt}
\begin{align*}
    \forall\lambda\in\text{spec}(\widetilde{\bm A}),\quad
\bm y^\top\widetilde{\bm A}=\lambda \bm y^\top,\, \bm y^\top\bm B=\bm 0\, \Rightarrow \bm y=0.
\end{align*}
\endgroup
Let $\bm y=\text{col}(\bm y_i)\in \mathbb{R}^{3n_H-1}$, where $\bm y_i=(y_{i}^\theta,y_{i}^\omega,y_{i}^v)$. From $\bm y^\top \bm B=0$, we obtain 
\begingroup
\setlength{\abovedisplayskip}{3pt}
\setlength{\belowdisplayskip}{5pt}
\begin{align*}
    y_{i}^\omega/m_i=0,\quad y_{i}^v/\tau_i=0.
\end{align*}
\endgroup
Since $m_i,\tau_i>0$, it follows that $y_{i}^\omega=y_{i}^v=0$ for all $i\in\mathcal{V}_H$. Substituting these formulas into $\bm y^\top \widetilde{\bm A}= \lambda \bm y^\top$ gives
\begingroup
\setlength{\abovedisplayskip}{3pt}
\setlength{\belowdisplayskip}{5pt}
\begin{align*}
    \bm y_i^\top(\bm A_{ii}+\bm A^{\bm x}_{i})+\sum_{(j,i)\in\mathcal{E}}\bm y_j^\top \widehat{\bm A}_{ji}+\sum_{(j,i)\in\mathcal{E}}\bm y_j^\top {\bm E}_{ji}=\lambda\bm y_i^\top.
\end{align*}
\endgroup
By the formulas of $\bm A_{ii}$, $\bm E$, $\bm A_i^{\bm x}$, and $\widehat{\bm A}_{ji}$ in \eqref{original}, we have $y_i^\theta=0$. Hence, $\bm y = \bm 0$, implying that $(\widetilde{\bm A}, \bm B)$ is controllable.
\par
\emph{Proof of (b).}
By the pole-placement theorem \cite[Th.~3.29]{Trentelman}, controllability of $(\widetilde{\bm A},\bm B)$
implies that for any prescribed set of $3n_H-1$ eigenvalues in
$\mathbb{C}^{-}$, there exists a feedback matrix $\bm F_c$ such that
\begin{equation}\label{central loop matrix}
  \text{spec}(\widetilde{\bm A}+\bm B\bm F_c)\subset\mathbb{C}^{-}.  
\end{equation}
\emph{Proof of (c).}
For such $\bm F_c$, the matrix $\widetilde{\bm A}+\bm B\bm F_c$ is Hurwitz. 
By Lyapunov's linearization method, $\widetilde{\bm x}=\bm 0$ of the
closed-loop nonlinear system is locally asymptotically stable.
Consequently, the synchronous-state trajectory
$\bm x^s(t)$ is locally asymptotically stable.
\hfill $\square$}
\par
Compared with the distributed controller of Case 2, the centralized controller computes feedback inputs based on the state information of all levels. {The centralized feedback matrix $\bm F_c$ incorporates coupling terms between adjacent levels, while the distributed matrix $\bm F_d$ has a block-diagonal structure.} However, such a centralized controller requires higher communication overhead and larger computational cost.
\section{Performance of closed-loop power system}\label{sec:performance}
While stability ensures that the state trajectories converge to the synchronous state under small disturbances, it is also important to measure the performance of the closed-loop power system. 
This section compares the distributed control method with
the centralized control method using the same linearized model and the
quadratic cost weights, under the closed-loop stability conditions
of Section~\ref{sec:closed-loop}.

\begin{defn}{\em Centralized LQR control problem of the linearized power system.} For the combined linearized model \eqref{original}, the centralized LQR control problem with states of all levels is formulated as follows,
    \begingroup
\setlength{\abovedisplayskip}{3pt}
\setlength{\belowdisplayskip}{5pt}
\begin{subequations}
    \begin{align}
  \frac{\text{d}\widetilde{\bm x}(t)}{\text{d}t}& =\widetilde{\bm A}\widetilde{\bm x}+\bm{B}\bm{u},\nonumber\\
  J_{c} &= \int_0^\infty
\begin{bmatrix}
\widetilde{\bm x}(t) \\
\bm u(t)
\end{bmatrix}^\top
\bm Q_{\text{cr}}
\begin{bmatrix}
\widetilde{\bm x}(t) \\
\bm u(t)
\end{bmatrix} \mathrm{d}t, \notag\\
\bm Q_{\text{cr}}&=\begin{bmatrix}
    \bm Q_{\bm x\bm x}&\bm 0\\
    \bm 0&\bm Q_{\bm u\bm u}
\end{bmatrix},\notag\\
\bm Q_{\bm x\bm x}&=\text{diag}(\bm Q_{\bm x\bm x,i}),~~\bm Q_{\bm u\bm u}=\text{diag}(\bm Q_{\bm u\bm u,i}),\notag\\
    0&=\!\widetilde{\bm Q}\widetilde{\bm A}\!+\!\widetilde{\bm A}^\top\widetilde{\bm Q}\!+\!\bm {Q}_{\bm{x}\bm{x}}-\widetilde{\bm Q}\bm{B}\bm{Q}_{\bm{u}\bm{u}}^{-1}\bm B^\top\widetilde{\bm Q},\label{central}\\
    \widetilde{\bm Q}&=\widetilde{\bm Q}^\top\succ 0,\nonumber \\ 
    \mathbb{C}^{-}&=\{c\in\mathbb{C}|\text{Re}(c)<0\},\nonumber \\
     \mathbb{C}^{-}&\supset \text{spec}({\widetilde{\bm{A}}}+\bm{B}\bm{F}_{c}),\notag\\
   \bm{F}_{c} &=-\bm{Q}_{\bm{u}\bm{u}}^{-1}\bm{B^\top}\widetilde{\bm Q},\label{cenfeed}\\       
    \bm{u}_{c}&=\bm{F}_{c}\widetilde{\bm{x}},\nonumber\\
    J_{c}^*&=\min J_{c}=\widetilde{\bm x}(0)^\top \widetilde{\bm Q}\widetilde{\bm x}(0),\label{cenper}
    \end{align}
\end{subequations}
\endgroup
where the control input $\bm{u}_{c}$ uses global feedback information to minimize the total quadratic cost $J_c$.
\end{defn}
The distributed feedback matrix $\bm F_d$ is obtained using local information in Section~\ref{sec:distributed_control}.
Under condition \eqref{closed-loop matrix}, the distributed
closed-loop matrix is Hurwitz, and its state trajectory is
\begingroup
\setlength{\abovedisplayskip}{3pt}
\setlength{\belowdisplayskip}{5pt}
\begin{equation}\label{solution}
\widetilde{\bm x}(t)
=e^{(\widetilde{\bm A}+\bm B\bm F_d)t}\widetilde{\bm x}(0).
\end{equation}
\endgroup
Its quadratic cost in the reduced coordinates is defined as
\begingroup
\setlength{\abovedisplayskip}{3pt}
\setlength{\belowdisplayskip}{5pt}
\begin{equation}\label{distributed-global-cost}
J_d=
\int_0^\infty
\widetilde{\bm x}(t)^\top
\left(
\bm Q_{\bm x\bm x}
+\bm F_d^\top\bm Q_{\bm u\bm u}\bm F_d
\right)
\widetilde{\bm x}(t)\,\mathrm{d}t.
\end{equation}
\endgroup
\begin{thm}{\em  Bounds on the difference of the quadratic cost.} \label{performancebound}
    Consider the centralized and distributed controllers evaluated
on the same reduced linearized model, with the same initial state
$\widetilde{\bm x}(0)$ and quadratic cost weights.
Assume that $\widetilde{\bm A}+\bm B\bm F_d$ is Hurwitz.
Let $J_d$ be defined by \eqref{distributed-global-cost}, and let
$J_c^*$ denote the centralized optimal cost. Then
\begingroup
\setlength{\abovedisplayskip}{3pt}
\setlength{\belowdisplayskip}{5pt}
    \begin{equation}\label{performance}
    \begin{aligned}
        0\le J_{d}\!-\!J_{c}^*&\le\frac{M}{\alpha}{\|(\bm F_{d}\!-\!\bm F_{c})^\top\bm{Q}_{\bm{u}\bm{u}}(\bm F_{d}\!-\!\bm F_{c})\|_2}\|\widetilde{\bm x}(0)\|^2,
    \end{aligned}
\end{equation}
\endgroup
where \(\bm S\succ 0\) is the solution to the Lyapunov equation
\begingroup
\setlength{\abovedisplayskip}{3pt}
\setlength{\belowdisplayskip}{3pt}
    \begin{align*}
(\widetilde{\bm A}+\bm B\bm F_d)^\top\bm S
+\bm S(\widetilde{\bm A}+\bm B\bm F_d)
=-\bm I,
    \end{align*}
    \endgroup
and
$
\alpha={1}/{\lambda_{\max}(\bm S)},~~
M={\lambda_{\max}(\bm S)}/
{\lambda_{\min}(\bm S)}.$
Here, \(\alpha>0\) and \(M\geq1\).
\end{thm}
\textit{{Proof: }} Define 
\begingroup
\setlength{\abovedisplayskip}{3pt}
\setlength{\belowdisplayskip}{3pt}
    \begin{align*}
            \widetilde{\bm{Q}}_{\text{cr}}=\bm Q_{\bm{x}\bm{x}}+
      \bm  F_{d}^\top \bm Q_{\bm{u}\bm{u}}\bm F_{d}.
    \end{align*}
    \endgroup
        Substituting the Riccati equation \eqref{central} and the feedback matrix \eqref{cenfeed} of the centralized LQR problem into the above formula, we obtain
        \begingroup
\setlength{\abovedisplayskip}{3pt}
\setlength{\belowdisplayskip}{5pt}
    \begin{equation*}
        \begin{aligned}
          \widetilde{\bm{Q}}_{\text{cr}}&= 
          -\widetilde{\bm Q}\widetilde{\bm A}\!-\!\widetilde{\bm A}^\top\widetilde{\bm Q}+\widetilde{\bm Q}\bm{B}\bm{Q}_{\bm{u}\bm{u}}^{-1}\bm B^\top\widetilde{\bm Q}+
      \bm F_d^\top \bm Q_{\bm{u}\bm{u}}\bm F_d\\
      &= -\widetilde{\bm Q}\widetilde{\bm A}\!-\!\widetilde{\bm A}^\top\widetilde{\bm Q}+\bm F_c^\top\bm{Q}_{\bm{u}\bm{u}}{\bm F_c}+
      \bm F_d^\top \bm Q_{\bm{u}\bm{u}}\bm F_d\\
      &=-\widetilde{\bm Q}(\widetilde{\bm A}+\bm B\bm F_d)\!-\!(\widetilde{\bm A}+\bm B\bm F_d)^\top\widetilde{\bm Q}\\
      &~~~+\bm F_c^\top\bm{Q}_{\bm{u}\bm{u}}\bm F_c+
      \bm F_{d}^\top \bm Q_{\bm{u}\bm{u}}\bm F_{d}+\widetilde{\bm Q}\bm B\bm F_{d}+\bm F_{d}^\top\bm B^\top\widetilde{\bm Q}\\
      &=-\widetilde{\bm Q}(\widetilde{\bm A}+\bm B\bm F_d)\!-\!(\widetilde{\bm A}+\bm B\bm F_d) ^\top\widetilde{\bm Q}\\
      &~~~+(\bm F_{d}-\bm F_{c})^\top\bm{Q}_{\bm{u}\bm{u}}(\bm F_{d}-\bm F_{c}).
        \end{aligned}
    \end{equation*} 
    \endgroup
 Using \eqref{distributed-global-cost}, we obtain
 \begingroup
\setlength{\abovedisplayskip}{3pt}
\setlength{\belowdisplayskip}{5pt}
    \begin{equation*}
        \begin{aligned}
           J_{d}&=\int_{0}^\infty
      \widetilde{\bm{x}}(t)^\top \widetilde{\bm{Q}}_{\text{cr}}\widetilde{\bm{x}}(t)\text{d}t\\
      &=-\int_{0}^\infty\frac{\text{d}}{\text{d}t}\bigg(\widetilde{\bm{x}}(t)^\top\widetilde{\bm Q}\widetilde{\bm{x}}(t)\bigg)\text{d}t+R_p\\
      &=\widetilde{\bm{x}}(0)^\top\widetilde{\bm Q}\widetilde{\bm{x}}(0)+R_p=J_{c}^*+R_p,
     \end{aligned}
    \end{equation*}
    \endgroup
where
\begingroup
\setlength{\abovedisplayskip}{3pt}
\setlength{\belowdisplayskip}{5pt}
\begin{align*}
    R_p = \int_{0}^\infty \widetilde{\bm x}(t)^\top
(\bm F_{d}-\bm F_{c})^\top \bm Q_{\bm{u}\bm{u}} (\bm F_{d}-\bm F_{c})
\widetilde{\bm x}(t)\, \mathrm{d}t,
\end{align*}
\endgroup
which represents the difference of the quadratic cost. In addition, the boundary term at infinity vanishes because the distributed
closed-loop matrix is Hurwitz. Because $\bm{Q}_{\bm{u}\bm{u}}\succ 0$, we have $R_p\ge 0$. 
Consider the Lyapunov function \(\widehat{W}( \widetilde{\bm x})= \widetilde{\bm x}^\top \bm S \widetilde{\bm x}\) for the closed-loop system \(\dot{\widetilde{\bm x}}=(\widetilde{\bm A}+\bm B\bm F_d) \widetilde{\bm x}\). Along its trajectories,
\begingroup
\setlength{\abovedisplayskip}{3pt}
\setlength{\belowdisplayskip}{5pt}
\begin{align*}
\dot{\widehat{W}}( \widetilde{\bm x})
&= \widetilde{\bm x}^\top\bigg((\widetilde{\bm A}+\bm B\bm F_d)^\top \bm S+\bm S(\widetilde{\bm A}+\bm B\bm F_d)\bigg) \widetilde{\bm x}\\
&=-\| \widetilde{\bm x}\|_2^2\le -\frac{1}{\lambda_{\max}(\bm S)}\widehat{W}( \widetilde{\bm x})=-\alpha \widehat{W}( \widetilde{\bm x}),
\end{align*}
\endgroup
where the inequality follows from
$
\widehat{W}( \widetilde{\bm x})\le \lambda_{\max}(\bm S)\| \widetilde{\bm x}\|_2^2.
$
For $t\ge0$, it follows that
\begingroup
\setlength{\abovedisplayskip}{3pt}
\setlength{\belowdisplayskip}{3pt}
    \begin{align*}
\widehat{W}(\widetilde{\bm x})\le e^{-\alpha t}\widehat{W}( \widetilde{\bm x}(0)).
    \end{align*}
    \endgroup
Since
\begingroup
\setlength{\abovedisplayskip}{3pt}
\setlength{\belowdisplayskip}{3pt}
    \begin{align*}
\lambda_{\min}(\bm S)\|\widetilde{\bm x}\|_2^2
\le \widehat{W}(\widetilde{\bm x})
\le\lambda_{\max}(\bm S)\|\widetilde{\bm x}\|_2^2,
    \end{align*}
    \endgroup
we obtain
\begingroup
\setlength{\abovedisplayskip}{3pt}
\setlength{\belowdisplayskip}{3pt}
    \begin{align*}
\|\widetilde{\bm x}(t)\|_2^2
&\le \frac{1}{\lambda_{\min}(\bm S)}\widehat{W}(\widetilde{\bm x}(t))\le \frac{e^{-\alpha t}}{\lambda_{\min}(\bm S)}\widehat{W}(\widetilde{\bm x}(0))\\
&\le \frac{\lambda_{\max}(\bm S)}{\lambda_{\min}(\bm S)}
e^{-\alpha t}\|\widetilde{\bm x}(0)\|_2^2=M e^{-\alpha t}\|\widetilde{\bm x}(0)\|_2^2.
\end{align*}
\endgroup
Then we calculate the following upper bound,
\begingroup
\setlength{\abovedisplayskip}{3pt}
\setlength{\belowdisplayskip}{5pt}
\begin{equation*}
    \begin{aligned}    
         R_p&\!\le\! M\|(\bm F_{d}\!-\!\bm F_{c})^\top\bm{Q}_{\bm{u}\bm{u}}(\bm F_{d}\!-\!\bm F_{c})\|_2
        \|\widetilde{\bm x}(0)\|^2\int_{0}^\infty e^{-\alpha t}\text{d}t\\
        &\!=\!\frac{M}{\alpha}{\|(\bm F_{d}\!-\!\bm F_{c})^\top\bm{Q}_{\bm{u}\bm{u}}(\bm F_{d}\!-\!\bm F_{c})\|_2}\|\widetilde{\bm x}(0)\|^2,
    \end{aligned}
\end{equation*}
\endgroup
which provides an upper bound on the performance difference between the distributed and centralized control methods.
\hfill $\square$
\par
The effects of the distributed and centralized control methods on the performance are quantified by bounds \eqref{performance}. The numerator quantifies the deviation of the feedback matrix between the distributed and the centralized methods, which is determined by the coupling terms $\bm E+\bm A^{\bm x} + \widehat{\bm A}$. 
If these coupling terms are small, and the weighting matrices are fixed, then the state matrix of the overall system is close to a block-diagonal form, and the feedback matrix of the distributed control approaches that of the centralized controller, i.e., $\bm F_d \approx \bm F_c$.
\par
The bound in \eqref{performance} depends on the transient response of the distributed closed-loop system. With the feedback gains and $M$ fixed, a larger 
$\alpha$ yields a smaller upper bound by reducing the accumulation of the additional cost. Weak inter-level coupling may lead to closer distributed and centralized feedback gains, which can reduce the upper bound. Since this worst-case estimate is derived using norm inequalities, it may be conservative and should be regarded as a theoretical bound rather than a tight estimate of the actual performance gap.
\section{Conclusion}\label{sec:conclusion}
This paper proposes a multilevel control system for a large-scale distribution network with high RES penetration, using ESSs as buffers at higher levels. 
Within this framework, a distributed LQR control strategy is developed to reduce computational cost and communication overhead caused by system interconnections.
{ A Lyapunov-based local stability analysis is provided for the open-loop power system, while the stability of the closed-loop system is established under sufficient conditions of the system matrix. In addition, an analytical bound is derived to quantify the performance gap relative to a centralized control method.
\par
Future work will focus on applying the proposed framework in a concrete distribution network in cooperation with a Distribution System Operator. The distributed control strategy will be evaluated under realistic network configurations and operational conditions to assess its practical applicability.}

\bibliographystyle{plain}        
\bibliography{ctrmultdisnet} 
\end{document}


\date{}
\title{\fontsize{18}{22}\selectfont
Control and Stability of a Multilevel Power System\\
for a Future Distribution Network\\[0.45em]
\large Examples and computations}
\author{Xian Wu}
\maketitle
\begin{abstract}
This report compares centralized and distributed control methods of three linearized power systems. Example 1 considers one subsystem at each of Levels 3, 4, and 5, which compares the distributed and centralized controllers in a minimal multilevel configuration. Example 2 has two Level-3 subsystems connected to one Level-4 subsystem. It compares the
closed-loop stability and the response to an initial perturbation at one
subsystem. Example 3 considers four Level-3 subsystems, two Level-4 subsystems, and one Level-5 subsystem.
It analyzes the stability in a larger configuration and compares the responses at
Level 5. Both controllers stabilize all three models. The distributed
controller has a higher quadratic cost in each example.
\end{abstract}

\section{The model and parameters}
\label{sec:settings}
\par
We use the multilevel power system in equation (1) of the paper. We linearize this system around a synchronous state and use node 1 as
the reference node. The centralized controller is designed from the combined
model (1). The distributed controller is designed from the isolated model
(13). Both controllers are evaluated on the same combined linearized model.
\par
The parameter values of the power system are listed in Tables~\ref{tab:simulation_parameters}, \ref{tab:example2_parameters}, and
\ref{tab:example3_parameters}. These values are chosen
for computation, but are not a detailed model of a physical
distribution network.
\par
The initial conditions of the power system 
are as follows. 
\[
\begin{aligned}
\widetilde{\delta}(0)&=(0,\,0.005,\,0.010)\ \mathrm{rad},\\
{\omega}(0)&=(0.010,\,0.010,\,0.010)\ \mathrm{Hz},\\
{v}(0)&=(0.020,\,-0.010,\,-0.010)\ \mathrm{p.u.}
\end{aligned}
\]
The nominal frequency and
voltage are 50 Hz and 1 p.u. Example 2 also uses two additional
initial conditions to analyze the interaction between subsystems.
\par
The cost function is the integral of a 
quadratic form of the state
deviations and control inputs, 
with weighing matrices $Q_{xx}=C^\top C$ and $Q_{uu}=I$. The diagonal entries
of $C$ are $(1,1)$ for the reference node and $(1/3,1,1)$ for each
non-reference node. The state order at a non-reference node $i$ is
$(\theta_i-\theta_1,\omega_i,v_i)$.
\par
The computations of a state trajectory of the linearized combined model use a time step of 0.01 s. Table~\ref{tab:simulation_settings}
lists the synchronous angles and computation intervals. 
\begin{table}[!htbp]
\centering
\caption{Synchronous angles and computation intervals}
\label{tab:simulation_settings}
\begin{tabular}{|l|c|c|}
\toprule
Example & Synchronous angles $\theta^\ast$ (rad) & Interval (s) \\
\hline
1 & $(0.10,\,0.10,\,0.10)$ & $[0,12]$ \\\hline
2 & $(0.10,\,0.11,\,0.12)$ & $[0,20]$ \\\hline
3 & $(0.10,\,0.11,\,0.12,\,0.10,\,0.20,\,0.10,\,0.10)$ & $[0,20]$ \\\hline
\end{tabular}
\end{table}

\section{Example 1: One Subsystem at Each Level}
Example 2 includes two Level-3 subsystems and one Level-4 subsystem. Subsystems $i = 1$
and $i = 2$ are at Level 3, and subsystem $i = 3$ is at Level 4. This configuration is selected to
investigate the interaction between the two Level-3 subsystems through their common Level-4
subsystem. The subsystem parameters are summarized in Table~\ref{tab:simulation_parameters}.
The active and reactive power injections are 
$
P=[0.4,0.3,0.4],~~
Q=[0.21, 0.72,0.34].
$
The conductance matrix $G$ and susceptance matrix $R$ are
\[
G=
\begin{bmatrix}
 0.5& -0.1&	0\\
-0.1&	0.5& -0.1\\
0& -0.1&	0.5 
\end{bmatrix},~~
R =
\begin{bmatrix}
0.2&	0.1&	0\\
0.1&	0.2&	0.1\\
0&	0.1&	0.2
\end{bmatrix}.
\]
\begin{table}[!htbp]
\centering
\caption{Subsystem parameters of Example 1}
\label{tab:simulation_parameters}
\begin{tabular}{|l|r|r|r|r|}
\hline
Subsystem  & $m_i$ & $d_i$ & $\tau_i$ & $k_i$ \\
\hline
$i=1$  & 0.2 & 0.1 & 0.1 & 0.51 \\
\hline
$i=2$ & 0.4 & 0.1 & 0.1 & 1.12 \\
\hline
$i=3$ & 0.6 & 0.1 & 0.1 & 0.64 \\
\hline
\end{tabular}
\end{table}
\begin{table}[!htbp]
\centering
\caption{Performance of the combined system at Level 5}
\label{tab:example1_performance}
\begin{tabular}{|l|r|r|r|}
\hline
\multirow{2}{*}{{Performance}} & \multirow{2}{*}{{Without control}} & \multicolumn{2}{c|}{{With control}} \\
            \cline{3-4}
             & & {Centralized}   & {Distributed}\\
            \hline
           \multirow{2}{*}{\shortstack[l]{Slowest\\eigenvalue}} & \multirow{2}{*}{0.436} & \multirow{2}{*}{-0.509} & \multirow{2}{*}{-0.801} \\
 & & & \\
\hline
           $\sigma_\omega$& 0.030 &0.001&0.001\\\hline
            $\sigma_v$ &0.114 & 0.004&0.001\\\hline
\end{tabular}
\end{table}
\par
Table~\ref{tab:example1_performance} shows the slowest eigenvalues of
the system matrices under different control strategies.
The uncontrolled model is unstable because its slowest
eigenvalue is positive. The centralized and distributed controllers give negative slowest eigenvalues, so both closed-loop systems are stable. The distributed controller gives a more negative slowest eigenvalue than the centralized controller.
The standard deviations of frequencies $\sigma_\omega$ and of voltages $\sigma_v$
are small under both controllers, while the distributed controller gives a smaller voltage standard
deviation.
\begin{figure}[H]
\centering
\includegraphics[width=0.85\linewidth]{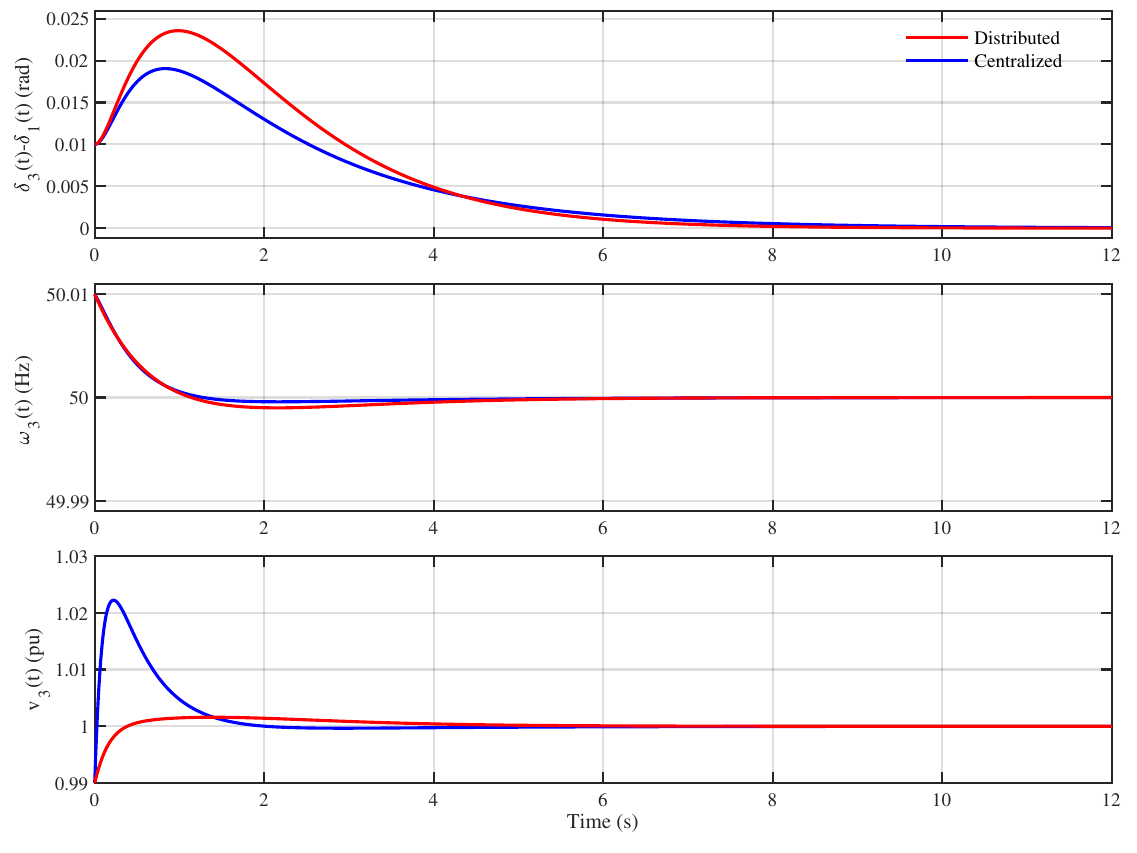}
\caption{Dynamic responses at Level 5 under distributed and centralized control}
\label{fig:example1_pcc_response}
\end{figure}
\par
Figure~\ref{fig:example1_pcc_response} shows the transient responses to simultaneous initial state deviations. Both controllers restore
the frequency and voltage to 50 Hz and 1 p.u. The voltage starts at
0.99 p.u. It rises above 1 p.u. and then returns to 1 p.u. Here, voltage
overshoot means the maximum excursion above the final value. The centralized controller gives a larger, earlier voltage peak. The
distributed controller gives a much smaller voltage overshoot, but a larger
phase-angle peak and a slight frequency undershoot. 
\par
Closed-loop stability guarantees decay, but it does not require each state to approach its final value
monotonically. The centralized controller minimizes the total quadratic
cost for the chosen weights. This objective includes states and control inputs. It does not minimize voltage overshoot alone.
Therefore, a larger voltage peak is compatible with its lower total cost
in Table~\ref{tab:cost_comparison}.
\section{Example 2: One Level-4 Subsystem with Two Level-3 Subsystems}
Example 2 includes two Level-3 subsystems and one Level-4 subsystem. Subsystems $i = 1$
and $i = 2$ are at Level 3, and subsystem $i = 3$ is at Level 4. This configuration is selected to
investigate the interaction between the two Level-3 subsystems through their common Level-4
subsystem. The subsystem parameters are summarized in  Table~\ref{tab:example2_parameters}.
The active and reactive power injections are 
$
P=[0.398,0.399,0.303],~~
Q=[0.232, 0.231,0.127].
$
The conductance matrix $G$ and susceptance matrix $R$ are
\[
G=
\begin{bmatrix}
 0.5&	0& -0.1\\
0&	0.5& -0.1\\
-0.1& -0.1&	0.5 
\end{bmatrix},~~
R =
\begin{bmatrix}
0.2&	0&	0.1\\
0&	0.2&	0.1\\
0.1&	0.1&	0.2
\end{bmatrix}.
\]
\begin{table}[!htbp]
\centering
\caption{Subsystem parameters of Example 2}
\label{tab:example2_parameters}
\begin{tabular}{|l|r|r|r|r|}
\hline
Subsystem  & $m_i$ & $d_i$ & $\tau_i$ & $k_i$ \\
\hline
$i=1$  & 0.1 & 0.2 & 0.1 & 0.53 \\
\hline
$i=2$ & 0.2 & 0.2 & 0.1 & 0.53 \\
\hline
$i=3$ & 0.3 & 0.2 & 0.3 & 0.53 \\
\hline
\end{tabular}
\end{table}
\begin{table}[!htbp]
\centering
\caption{Eigenvalues of the linearized combined system in Example 2}
\label{tab:example2_eigenvalues}
\begin{tabular}{|r|r|r|}
\toprule
\makecell[r]{Without\\control} &  \makecell[r]{Centralized\\control} &   \makecell[r]{Distributed\\control} \\
\hline
  0.832& $-0.464$& $-0.616$\\\hline
  0.344& $-0.649$& $-1.151$\\\hline
  0.344& $-3.520$& $-1.151$\\\hline
  0.263& $-3.520$& $-3.180$\\\hline
  0.263& $-8.055$& $-4.414$\\\hline
 $-0.840$& $-8.055$& $-6.073$\\\hline
 $-2.213$& $-10.868$& $-6.831$\\\hline
 $-3.046$& $-10.868$& $-10.338$\\\hline
\end{tabular}
\end{table}
\par
Table~\ref{tab:example2_eigenvalues} compares the real parts of all eigenvalues
in descending order. The uncontrolled model is unstable because its slowest eigenvalue is positive. Under centralized and distributed
control, the slowest eigenvalues are -0.464 and -0.616. Both closed-loop matrices
are Hurwitz. Thus, the distributed controller stabilizes the combined
model using
local state information.
\par
To examine the interaction between the two Level-3 subsystems, we additionally consider
different initial conditions. 
\[
\begin{aligned}
\widehat{\omega}(0)&=(0.010,\,0.002,\,0.010)\ \mathrm{Hz},\\
\widehat{v}(0)&=(0.020,\,0.003,\,-0.010)\ \mathrm{p.u.}
\end{aligned}
\]
The initial phase angle difference deviations remain the same as the previous example.
\begin{table}[!htbp]
\centering
\caption{Maximum gaps between subsystems 1 and 2 
with different initial deviations}
\label{tab:example2_interaction}
\begin{tabular}{|l|c|c|}
\hline
Control method &
\makecell{Maximum absolute\\frequency gap (Hz)} &
\makecell{Maximum absolute\\voltage gap (p.u.)}\\
\hline
Without control & 382.900083 & 2982.344275 \\
\hline
Centralized control & 0.008 & 0.017 \\
\hline
Distributed control & 0.008 & 0.017  \\
\hline
\end{tabular}
\end{table}
\par
Table~\ref{tab:example2_interaction} shows the maximum absolute frequency and voltage gaps between subsystems 1 and
2 over 0--20 s. The gaps are defined as the absolute differences of the frequencies and voltages of these two subsystems.
The initial gaps are 0.008 Hz and
0.017 p.u. Under both controllers, the maximum gaps occur
at $t=0$. The large uncontrolled values result
from extending an unstable linearized model over 20 s.
\par
We next perturb only the initial frequency of subsystem 1 by 0.010 Hz at $t=0$.
All other initial state deviations are zero. Table~\ref{tab:example2_single_disturbance}
lists their maximum absolute frequency deviations over 0--20 s.
\begin{table}[!htbp]
\centering
\caption{Maximum frequency responses under an initial frequency
perturbation}
\label{tab:example2_single_disturbance}
\begin{tabular}{|l|c|c|}
\hline
Control method &
\makecell{Subsystem 2} &
\makecell{Subsystem 3} \\
\hline
Centralized control & 0.0001 & 0.0001 \\
\hline
Distributed control & 0.0002 & 0.0003 \\
\hline
\end{tabular}
\end{table}
\par
Both initially unperturbed subsystems develop nonzero frequency deviations.
The perturbation propagates beyond subsystem 1 in the closed-loop
model. 
Centralized control gives the smaller peaks in this example.
Therefore, both electrical coupling and
controller methods affect the perturbation.
\section{Example 3: One Level-5 Subsystem with Two Level-4 Subsystems and Four Level-3 Subsystems}
Example 3 considers a three-level system with seven subsystems. Subsystems $i = 1$, 2, 3, 4 belong
to Level 3, subsystems $i = 5$, 6 belong to Level 4, and subsystem $i = 7$ belongs to Level 5.
Subsystems 1 and 2 are connected to subsystem 5, subsystems 3 and 4 are connected to subsystem
6, and subsystems 5 and 6 are connected to subsystem 7. This example is intended to investigate
a larger three-level power system and to explore how the interconnected subsystems interact
across and within the levels. In contrast to Example 1, its branched configuration contains
multiple subsystems at Levels 3 and 4. The subsystem parameters are summarized in Table~\ref{tab:example3_parameters}. The active and reactive power injections are 
\begin{align*}
    P&=[ 0.391,0.391,0.402,0.400,0.230,0.198,0.291],\\
Q&=[0.210, 0.209,0.198,0.200,-0.028,0.002,0.110].
\end{align*}
The conductance matrix $G$ and susceptance matrix $R$ are
\begingroup
\setlength{\arraycolsep}{3pt}
\[
G=
\begin{bmatrix}
 0.5&	0& 0&0&-0.1&0&0\\
0&	0.5&	0&0&-0.1&0&0\\
0&	0&	0.5 &0&0&-0.1&0\\
0 &0 & 0& 0.5& 0& -0.1 & 0\\
-0.1&-0.1&0&0&0.5&0&-0.1\\
0&0&-0.1&-0.1&0&0.5&-0.1\\
0&0&0&0&-0.1&-0.1&0.5
\end{bmatrix},\qquad
R=
\begin{bmatrix}
0.2&	0&	0&0&0.1&0&0\\
0&	0.2&0&0&	0.1&0&0\\
0&0&0.2&0&0&	0.1&0\\
0&0&0&0.2&0&0.1&0\\
0.1&0.1&0&0&0.2&0&0.1\\
0&0&0.1&0.1&0&0.2&0.1\\
0&0&0&0&0.1&0.1&0.2
\end{bmatrix}.
\]
\endgroup
\begin{table}[!htbp]
\centering
\caption{Subsystem parameters of Example 3}
\label{tab:example3_parameters}
\begin{tabular}{|l|r|r|r|r|}
\hline
Subsystem  & $m_i$ & $d_i$ & $\tau_i$ & $k_i$ \\
\hline
$i=1$  & 0.1 & 0.2 & 0.1 & 0.5 \\
\hline
$i=2$ & 0.1 & 0.2 & 0.1 & 0.5 \\
\hline
$i=3$ & 0.1 & 0.2 & 0.1 & 0.5 \\
\hline
$i=4$ & 0.1 & 0.2 & 0.1 & 0.5 \\
\hline
$i=5$ & 0.2 & 0.2 & 0.2 & 0.5 \\
\hline
$i=6$ & 0.2 & 0.2 & 0.2 & 0.5 \\
\hline
$i=7$ & 0.3 & 0.2 & 0.3 & 0.5 \\
\hline
\end{tabular}
\end{table}
\begin{table}[!htbp]
\centering
\caption{Performance of the combined system at Level 5}
\label{tab:example3_performance}
\begin{tabular}{|l|r|r|r|}
\hline
\multirow{2}{*}{{Performance}} & \multirow{2}{*}{{Without control}} & \multicolumn{2}{c|}{{With control}} \\
            \cline{3-4}
             & & {Centralized}   & {Distributed}\\
            \hline
           \multirow{2}{*}{\shortstack[l]{Slowest\\eigenvalue}} & \multirow{2}{*}{1.776} & \multirow{2}{*}{-0.338} & \multirow{2}{*}{-0.408} \\
 & & & \\
\hline
           $\sigma_\omega$& 189135336565.456 &0.001&0.001\\\hline
            $\sigma_v$ & 1665122911975.300& 0.001&0.001\\\hline
\end{tabular}
\end{table}
\par
Table~\ref{tab:example3_performance} shows the slowest eigenvalue of
the system matrices. The uncontrolled model is unstable because its slowest eigenvalue is positive.
The centralized and distributed controllers have the slowest eigenvalues of
$-0.338$ and $-0.408$. Both stabilize the combined linearized
model. The very large uncontrolled standard deviations reflect divergence
of the linearized model over 20 s.
Under both controllers, the frequency and voltage standard deviations round
to 0.001.
\begin{figure}[H]
\centering
\includegraphics[width=0.85\linewidth]{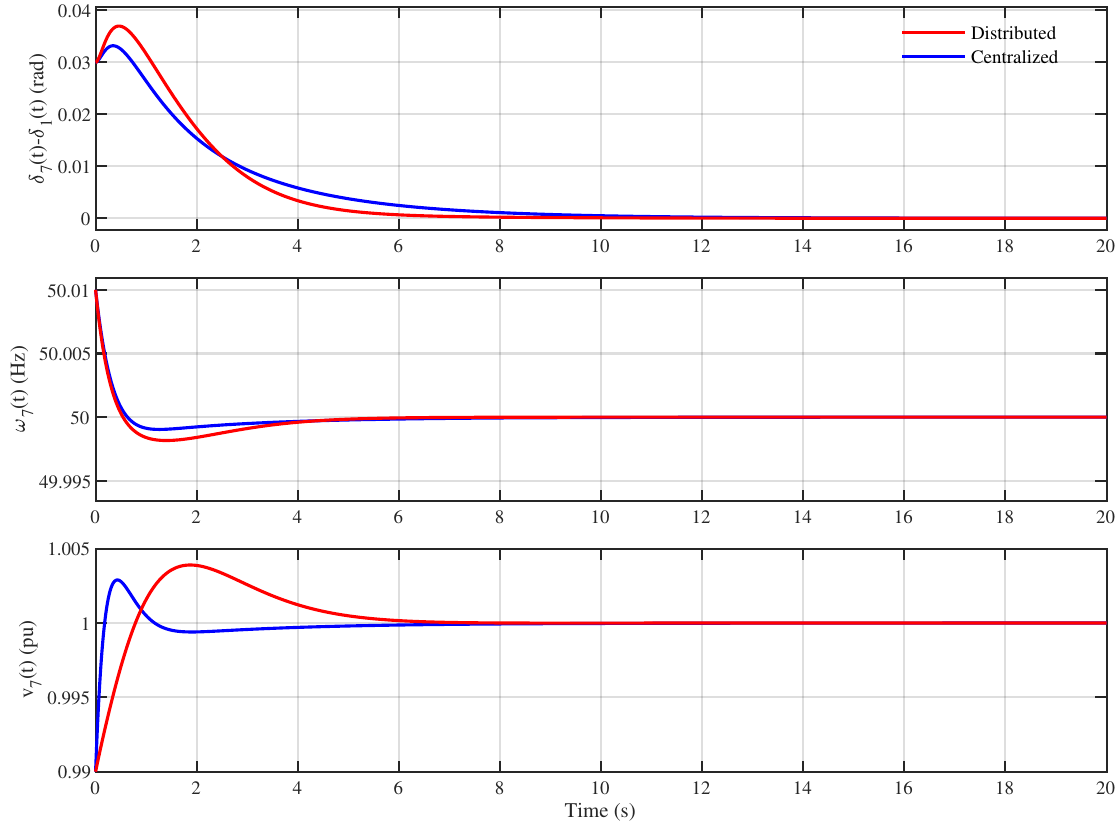}
\caption{Dynamic responses at Level 5 under distributed and centralized control}
\label{fig:example3_pcc_response}
\end{figure}
\par
Figure~\ref{fig:example3_pcc_response} shows that the phase angle difference
returns to its synchronous value. The frequency and voltage return to
50 Hz and 1 p.u. The distributed controller gives a larger phase-angle
peak, a deeper frequency undershoot, and a larger voltage overshoot.
Its voltage peak also occurs later. 
\par
This example shows that both controllers stabilize the linearized
combined system. 
\section{Comparisons
of the Examples}
\begin{table}[!htbp]
\centering
\caption{Quadratic-cost comparison of the combined linearized models}
\label{tab:cost_comparison}
\begin{tabular}{|c|r|r|r|r|}
\hline
Example & $J_c^*$ & $J_d$ &
$J_d-J_c^*$ &
Theoretical Bound \\
\hline
1 & 0.0040 & 0.0048 &0.0008 &  0.5722 \\
\hline
2 & 0.0021 & 0.0023& 0.0002 & 0.4487 \\
\hline
3 & 0.0045 & 0.0054 & 0.0009 &   3.1490\\
\hline
\end{tabular}
\end{table}
\par
Table~\ref{tab:cost_comparison} compares the centralized optimal cost
$J_c^*$ with the distributed cost $J_d$. The distributed cost is higher than the centralized controller in all three examples.
Each cost gap satisfies the theoretical bound of Theorem 29 of the manuscript, although the bounds are conservative for these parameter choices.
\section{Conclusions}
Both controllers stabilize all three combined linearized models. The
distributed controller gives a more negative slowest eigenvalues. It has a smaller voltage overshoot in Example 1, but it does not
give smaller peaks in every state. In addition, the distributed controller achieves performance comparable
to that of the centralized controller in these examples.